%% file: main.tex
\documentclass[a4paper,11pt]{article}
\pdfoutput=1
\usepackage[absolute,overlay]{textpos} 
\usepackage{tcs}
\usepackage{tikz}
\usepackage{dsfont}
\usepackage[font=footnotesize]{caption}

\DeclareUnicodeCharacter{1EF3}{\'{y}}
\usepackage{tocloft}

\usepackage{etoolbox}
\newbool{journal}
\setbool{journal}{false} 

\begin{document}

\setlength{\TPHorizModule}{1cm}
\setlength{\TPVertModule}{1cm}
\begin{textblock}{5}(13.4,1.5)
    \noindent \raggedleft \footnotesize{MIT-CTP/6005}
\end{textblock}


\title{
Entanglement in quantum spin chains \\is strictly finite at any temperature}
\author{
Ainesh Bakshi \\
\texttt{ainesh@nyu.edu} \\
NYU
\and
Soonwon Choi \\
\texttt{soonwon@mit.edu} \\
MIT
\and
Sa\'ul Pilatowsky-Cameo \\
\texttt{saulpila@mit.edu} \\
MIT
}
\date{}

\maketitle

\begin{abstract}
\input{abstract}

\end{abstract}

\thispagestyle{empty}
\clearpage
\newpage

\microtypesetup{protrusion=false}
{
  \hypersetup{linkcolor=black}
\tableofcontents{}
}
\thispagestyle{empty}
\microtypesetup{protrusion=true}
\clearpage
\setcounter{page}{1}

\input{intro}
\input{prelims}

\input{separability}

\input{combinatorics}
\input{entanglement-bulk-decomp}

\input{prep_algorithm}

\section*{Acknowledgments}
\addcontentsline{toc}{section}{Acknowledgments}
\input{acknowledgments}

\printbibliography

\end{document}

%% file: abstract.tex
Entanglement is the hallmark of quantum physics, yet its characterization in interacting many-body systems at thermal equilibrium remains one of the most important challenges in quantum statistical physics. We prove that the Gibbs state of any quantum spin chain can be exactly decomposed into a mixture of matrix product states with a bond dimension that is independent of the system size, at any nonzero temperature. As a consequence, the Schmidt number, arguably the most stringent measure of bipartite entanglement, is strictly finite for thermal states, even in the thermodynamic limit. Our decomposition is explicit and is accompanied by an efficient classical algorithm to sample the resulting matrix product states.

%% file: intro.tex
\newcommand{\dyad}[1]{|#1\rangle \langle #1|}

\definecolor{mpstensors}{RGB}{163, 205, 255}
\definecolor{Mmatrices}{RGB}{255, 195, 106}
\definecolor{statecircs}{RGB}{255, 255, 255}
\definecolor{gibbsstatecolor}{RGB}{255, 153, 0}
\definecolor{Xboxcolor}{RGB}{220, 220, 220}

\newcommand\tensorheight{0.8}
\newcommand\shorttensorheight{0.7}
\newcommand\tensorwidth{0.9}

\newcommand\shorttensorstaking{2*\shorttensorheight} 
\pgfdeclarelayer{bg}    
\pgfsetlayers{bg,main}  

\newcommand\tensorlegheight{1.3}
\newcommand\shorttensorlegheight{1.1}

\newcommand{\drawbox}[4]{
      \filldraw[fill=#3] (-#1,-#2) -- (-#1,#2) -- (#1,#2) -- (#1,-#2) -- (-#1,-#2);
        \draw (0,0) node {\small #4};
}

\newcommand{\drawboxnoright}[4]{
      \filldraw[fill=#3] (#1,-#2) -- (-#1,-#2) -- (-#1,#2) -- (#1,#2); 
        \draw (0,0) node {\small #4};

}
\newcommand{\tensorbox}[1]{ \drawbox{\tensorwidth}{\tensorheight}{mpstensors}{#1}}
\newcommand{\drawMbox}[2]{
\begin{scope}[shift={(#1)}]
      \drawbox{\Mboxwidth}{\shorttensorheight}{Mmatrices}{#2} 
 	\end{scope}
}

\newcommand{\drawgibbs}[3]{
\begin{scope}[shift={(#1)}]
      \drawboxnoright{#2}{\shorttensorheight}{gibbsstatecolor}{#3} 
 	\end{scope}
}
\newcommand{\drawXbox}[3]{
\begin{scope}[shift={(#1)}]
      \drawboxnoright{#2}{\shorttensorheight}{Xboxcolor}{#3} 
 	\end{scope}
}
\newcommand{\lefttensor}[2]{
	\begin{scope}[shift={(#1)}]
		\draw (0,0) -- (\tensorlegheight,0);
		\draw (0,\tensorlegheight) -- (0,0); 
        \tensorbox{#2};
	\end{scope}
}
\newcommand{\middletensor}[2]{
	\begin{scope}[shift={(#1)}]
		\draw (-\tensorlegheight,0) -- (\tensorlegheight,0);
		\draw (0,\tensorlegheight) -- (0,0);  
        \tensorbox{#2};
	\end{scope}
}
\newcommand{\righttensor}[2]{
	\begin{scope}[shift={(#1)}]
		\draw (-\tensorlegheight,0) -- (0,0);
		\draw (0,\tensorlegheight) -- (0,0); 
        \tensorbox{#2};
	\end{scope}
}
\newcommand{\lefttensordag}[2]{
	\begin{scope}[shift={(#1)}]
		\draw (0,0) -- (\tensorlegheight,0);
		\draw (0,-\tensorlegheight) -- (0,0); 
        \tensorbox{#2};
	\end{scope}
}
\newcommand{\middletensordag}[2]{
	\begin{scope}[shift={(#1)}]
		\draw (-\tensorlegheight,0) -- (\tensorlegheight,0);
		\draw (0,-\tensorlegheight) -- (0,0);  
        \tensorbox{#2};
	\end{scope}
}
\newcommand{\righttensordag}[2]{
	\begin{scope}[shift={(#1)}]
		\draw (-\tensorlegheight,0) -- (0,0);
		\draw (0,-\tensorlegheight) -- (0,0); 
        \tensorbox{#2};
	\end{scope}
}
\pgfmathsetmacro\MathAxis{height("$\vcenter{}$")}

\newcommand{\marginout}[0]{0.5}
\newcommand{\Mboxwidth}[0]{2.5}
\newcommand{\drawingsscale}[0]{0.3}

\section{Introduction}
Entanglement is the key resource that distinguishes quantum systems from their classical counterparts. In quantum information, entanglement enables tasks that are otherwise impossible, including quantum teleportation~\cite{Bennett1993}, superdense coding~\cite{Benett1992}, device-independent cryptography~\cite{Acin2007}, and highly sensitive quantum metrology~\cite{Degen2017, Giovannetti2011}. In many-body physics, entanglement characterizes 
topological phases of matter~\cite{Wen2017,Shi2020}, and its growth and spread in quenched dynamics govern thermalization and information scrambling~\cite{Gogolin2016}. In equilibrium, most quantum many-body systems are well-described by the thermal Gibbs state, parametrized by a few macroscopic quantities such as temperature.  Therefore, a foundational question that cuts across quantum information, computation and many-body physics is as follows:

\begin{quote}
\begin{center}
\emph{How much entanglement is contained in the thermal state of a quantum system?}
\end{center}
\end{quote}

Entanglement in mixed states, such as the thermal state, is notoriously hard to quantify \cite{Horodecki2009}. Unlike  pure states, whose entanglement is uniquely characterized by their subsystem entropy, mixed states can be decomposed into many different ensembles of pure states, each exhibiting widely different amounts of entanglement. The standard approach is then to seek the \emph{optimal} unraveling into pure states, one where the constituent states have low entanglement. This optimization task is, however, rendered intractable by the fact that the underlying Hilbert space grows exponentially in the system size. For example, even deciding whether the state is \emph{separable} (zero entanglement) is computationally hard~\cite{Gurvits2003}. One notable exception is the recent result of Bakshi, Liu, Moitra, and Tang~\cite{bakshi2024high}, which proves that thermal states of local Hamiltonians are entirely separable above a fixed constant temperature. A similar result for high-temperature Gibbs states of fermionic Hamiltonians was subsequently obtained by Ramkumar, Cai, Tong and Jiang~\cite{ramkumar2025high}.

In general, rigorously quantifying entanglement in thermal states below the separability threshold of~\cite{bakshi2024high} remains extremely challenging, even for one-dimensional systems. From a resource-theoretic standpoint, this difficulty can be ascribed to a lack of \emph{asymptotic reversibility}. 
Specifically, let us imagine asking how much entanglement is needed or can be extracted from multiple copies of a given quantum state.
For this purpose, given a bipartition of the state, local operations on each partition, as well as classical communication between the two parties, are considered free.
For pure states, a seminal result of Bennett, Bernstein, Popescu, and Schumacher~\cite{Bennett1996} provides an operational interpretation of von Neumann entropy and shows that it coincides with the number of Bell pairs required to prepare the state under local operations (\emph{entanglement cost}) as well as the number of Bell pairs that can be extracted from a state (\emph{distillable entanglement}), as the number of copies tends to infinity. Therefore, every pure state can be converted into and from a finite number of copies of maximally entangled Bell states at a unique rate. Mixed states, however, can possess {\it bound entanglement}, having an entanglement cost which is strictly larger than the distillable entanglement~\cite{Horodecki1998,Vidal2001}. Therefore, there is no single canonical entanglement measure for mixed states.   

While quantifying entanglement in general mixed states is computationally intractable, correlations in Gibbs states have been extensively studied by the physics community. In particular, it is widely believed that for thermal states, the amount of correlations is dictated by their locality. A prominent example is the thermal area law derived by Wolf, Verstraete, Hastings, and Cirac~\cite{Wolf2008}, which roughly states that the quantum mutual information across a subsystem scales proportional to its boundary\footnote{Wolf et al. prove an area law for local Hamiltonians in any spatial dimension and show the quantum mutual information scales linearly in the inverse temperature. This scaling was improved to sub-linear by~\cite{Kuwahara2021}.}. While this result confirms the expectation that correlations should be local, it offers limited utility in quantifying quantum entanglement. The quantum mutual information fails to distinguish between classical and quantum correlations; furthermore, it can be negligible even when large amounts of entanglement are present~\cite{Hayden2006}. Consequently, deeper insight is required to specifically characterize the quantum properties of systems in thermal equilibrium.

One-dimensional systems are unique because their entanglement can be canonically quantified using the matrix-product-state (MPS) representation, as formalized by Vidal~\cite{Vidal2003,PerezGarcia2007}. Because the boundary of any connected bipartition in 1D is a single point, it must mediate all the entanglement between the two sides. A low-entangled pure state can therefore be represented by a sequence of tensor contractions, where the \textit{bond dimension} directly upper bounds the entanglement entropy. Due to the efficiency and numerical stability of MPS algorithms for simulating ground states~\cite{White1992,Schollwoeck2011} and finite-temperature states~\cite{Verstraete2004MPDO,Stoudenmire2010}, they have become the workhorse representation in condensed-matter physics and quantum chemistry~\cite{Schollwick2005,Chan2011}. This formalism extends to mixed states, where entanglement can be bounded by expressing the system as a mixture of MPSs or as matrix product operators (MPOs). Despite increasing efforts to provide such approximate decompositions for thermal states~\cite{Berta2018,Kuwahara2021,Hastings2006,Pirvu2010,Eisert2014,GuthJarkovsk2020,Huang2021,Alhambra2021},  existing approximation theorems by themselves do not entirely quantify the entanglement in thermal states: in existing constructions the required bond dimension diverges both in the thermodynamic limit (as system size grows) and as the approximation error is driven to zero. 

\section{Results}
In this work, we provide an exact decomposition for the Gibbs state of any local 1D spin-chain Hamiltonian as a mixture of MPSs, where the bond dimension is independent of the system size. Such a decomposition places a strong bound on the amount of entanglement. Specifically, it implies that the {\it Schmidt number} across any cut is finite, even in the thermodynamic limit. We also provide a classical algorithm to sample from the distribution of MPSs in polynomial time. Taken together, our results demonstrate that entanglement in one-dimensional thermal states is strictly finite and establish mixtures of MPSs as an exact and classically accessible description of thermal equilibrium. 

\subsection{Thermal states are exact mixtures of matrix product states} 
  We study a spin chain of $\qubits$ qubits labeled from $1$ to $\qubits$. Low-entangled pure states of such a system can be efficiently represented as an MPS. For each possible state of the $j$th qubit, $\ket{b_j}\in\{\ket{0},\ket{1}\}$, we  consider a matrix $A_j^{\smash{b_j}}$ and build an MPS via matrix multiplication:
\begin{align}
\label{eq:MPSdef}
    \ket{\phi}&=\frac{1}{\sqrt{\alpha}}\sum_{b_j\in\{0,1\}} A_1^{b_1}A_2^{b_2}\cdots A_\qubits^{b_\qubits} \ket{b_1 b_2\cdots  b_\qubits},
    \end{align}
where the number $\alpha>0$ ensures normalization. For $1<j<\qubits$, the matrices $A_j^{\smash{b_j}}$ have dimension $\chi\times \chi$, while $A_1^{\smash{b_1}}$ and $A_\qubits^{\smash{b_\qubits}}$ are $1\times \chi$ and $\chi \times 1$-dimensional, respectively. The parameter $\chi$ is called the {\it bond dimension}, and it strictly bounds the amount of entanglement in $\ket{\phi}$.  Specifically, $\chi$ upper bounds the Schmidt rank of $\ket{\phi}$, and hence $\log(\chi)$ bounds the von Neumann entanglement entropy across any connected bipartition.

We study the  Gibbs state of a geometrically $\locality$-local Hamiltonian  ($\locality\geq 2$) in one dimension,
\begin{align}
    g_\beta&=\exp(-\beta H)/Z_\beta, & Z_\beta&=\tr(\exp(-\beta H)),
\end{align}
where $H=\sum_{j=1}^{\qubits-\locality+1} H_{j}$ and each local term $H_{j}$ has support between qubits $j$ and $j+\locality-1$ with bounded operator norm $\norm{H_j}\leq 1$.
The bound on the operator norm implicitly defines the units of temperature, which we treat as a dimensionless quantity henceforth.  The parameter $\beta$ is the inverse temperature, which we will assume to be some arbitrary constant. We take open boundary conditions for simplicity (no interactions between the first and last qubits), but our methods can be generalized to allow for periodic boundary conditions. Our main result guarantees that $g_\beta$ can be expressed exactly as a mixture of MPSs with constant bond dimension, at any nonzero temperature. Below, we denote $\tilde\beta\coloneqq \max\{1,\beta\}$.

\begin{theorem}[1D thermal states are exact mixtures of MPSs]
\label{th:GibbsStatesareMixturesOfMPS}
For any geometrically $\locality$-local Hamiltonian on a spin chain of length $\qubits$ and at any inverse temperature $0 \leq \beta<\infty$, the Gibbs state $g_\beta$ admits the following decomposition  (see Fig.~\ref{fig:01}):
\begin{equation}
    \label{eq:gibbsstatedecomp}
    g_\beta  =  \Paren{ M_1 \cdot M_2 \cdots M_{\qubits-m +1}} \cdot \sigma\cdot  ( M_1 \cdot M_2 \cdots M_{\qubits-m+1})^\dagger
\end{equation}
where each $M_i$ is supported between qubits $i$ and $i+m-1$, for $m=\exp(c\tilde{\beta})$ with $c = 10 \locality$ and $\sigma$ is separable over stabilizer product states\footnote{Stabilizer product states are of the form $\ket{\bm s}=\ket{s_1}\otimes \cdots \otimes\ket{s_\qubits}$, where each label $s_j\in\mathcal{S}\coloneqq\{0,1,+,-,+i,-i\}$ identifies an eigenstate of one of the Pauli matrices. }. In particular, Eq.~\eqref{eq:gibbsstatedecomp} implies that $g_\beta$ can be exactly decomposed as a mixture
\begin{equation}
\label{eq:mixtureMPS}
    g_\beta=\sum_{\bm s\in \mathcal{S}^\qubits}  p_{\bm  s} \dyad{\phi_{\bm  s}}
\end{equation}
of MPSs $\ket{\phi_{\bm  s}}$ of bond dimension $\chi\leq \exp(\exp(c\tilde \beta))$.  
\end{theorem}

 \begin{figure}[tb]
 \centering
\includegraphics{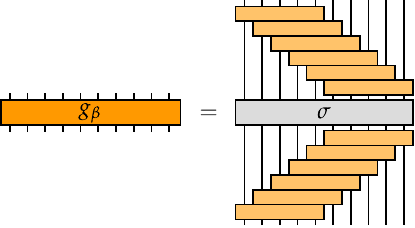}
\caption{Tensor network diagram of the decomposition of \cref{th:GibbsStatesareMixturesOfMPS}, where the Gibbs state $g_\beta$ is decomposed as a ladder of local operators acting on a separable state $\sigma$. Each vertical line represents a qubit, and each box represents an operator supported within. Vertically stacking boxes represents tensor contraction or equivalently matrix multiplication.}
\label{fig:01}
 \end{figure}

There exists a rich body of literature that rigorously constructs tensor network approximations to one-dimensional Gibbs states. Huang~\cite{Huang2021} and Alhambra and Cirac~\cite{Alhambra2021} obtain local approximations to Gibbs states via an MPO, where the approximation degrades exponentially in the locality. Kliesch, Gogolin, Kastoryano, Riera and Eisert~\cite{Eisert2014} obtain a global approximation in the high-temperature regime. Berta, Brand\~ao, Haegeman, Scholz and Verstraete~\cite{Berta2018} obtained the first global MPS approximation that works at any finite temperature. Kuwahara, Alhambra and Anshu~\cite{Kuwahara2021} have the best global approximation, where the bond dimension scales sub-exponentially in $\beta$ and sub-linearly in system size $\qubits$ and the inverse accuracy $1/\eps$.   To the best of our knowledge, all prior work achieves approximate tensor-network decompositions, where the bond dimension grows with system size and the desired closeness to the target Gibbs state.  In contrast, our MPS mixture is an exact decomposition of the Gibbs state, with bond dimension that scales doubly exponentially in $\beta$, but is independent of system size.   
See~\cref{table1} for a detailed comparison.

\begin{table*}[tb]
    \centering
\begin{tabular}{ |c|c|c|c| } 
 \hline
 Reference& Type & Accuracy ($\varepsilon$) & Bond dimension \\ 
 \hline
  
    \cite{Huang2021}& MPO & Local  &  $\exp(\widetilde{\mathcal{O}}(\beta^{2/3}+\sqrt{\beta \log(1/\varepsilon)}))$ \\ 
    \cite{Alhambra2021}& MPO & Local  &  $(\xi/\eps) \exp(\widetilde{\mathcal{O}}(\max(\beta,\sqrt{\beta \log(\xi/\eps)})))$ \\ 
    \cite{Eisert2014}& MPO & Global  &  $\exp(\bigO{\ln({C(\beta)\qubits/\varepsilon})})$ (high temperature) \\ 
    \cite{Molnar2015}& MPO & Global  &  $(\qubits/ \varepsilon)^{\bigO{\beta}}$ \\ 
   \cite{Kuwahara2021}& MPO & Global   & $\exp(q\log(q)),  q=\bigO{\max(\beta^{2/3},[\beta\log(\beta \qubits/\varepsilon)]^{1/2})}$ \\ 
 \cite{Berta2018}& MPS  & Global  & $ 
\exp(e^{\mathcal{O}(\beta)}\log^2(n/\varepsilon))
 $ \\ 
  \cite{Kuwahara2021}& MPS & Global   & $\exp(\widetilde{\mathcal{O}}(\max(\beta^{2/3},[\beta\log(\beta \qubits/\varepsilon)]^{1/2})))$ \\ 
 \cref{th:GibbsStatesareMixturesOfMPS} & MPS  & Global, exact & $\exp(\exp( \mathcal{O}( \beta )))$  \\ 
 \hline
\end{tabular}
\caption{Bond dimension for tensor network approximations to thermal states in 1D using matrix product operators (MPOs) or mixtures of matrix product states (MPSs), up to global or local trace distance $\varepsilon$, as specified. The symbol $\xi$ denotes correlation length, $\qubits$ number of qubits, and $\beta$ inverse temperature.}
\label{table1}
\end{table*}
\subsection{Bound on thermal entanglement in 1D} 
Theorem~\ref{th:GibbsStatesareMixturesOfMPS} places a strong limit on the bipartite quantum entanglement of the thermal state, as measured by its \emph{Schmidt number}.  The Schmidt number of a mixed state, introduced by Terhal and Horodecki~\cite{Terhal2000}, bounds the Schmidt rank of all pure states in an optimal unraveling. This is arguably the most stringent measure of bipartite entanglement in mixed states, since its logarithm bounds many well-studied entanglement measures~\cite{Horodecki2009}, including the R\'enyi entanglement of formation of every order~\cite{Kuwahara2021} and the entanglement of distillation~\cite{Bennett1996} (see \cref{sec:prelims} for formal definitions). Theorem~\ref{th:GibbsStatesareMixturesOfMPS} immediately implies a finite Schmidt number for the Gibbs state in 1D systems.

\begin{corollary}[Entanglement is strictly finite]
Given any inverse temperature $0\leq \beta < \infty$ and any $\qubits$-qubit, one-dimensional local Hamiltonian $H$, the Gibbs state $g_\beta$ has Schmidt number at most $$\mathrm{SN}({g_\beta}) \leq \bonddim\leq \exp(\exp\paren{ c\tilde{\beta}})$$ across every contiguous bipartition of the system, independently of the system size.
\end{corollary}

As an immediate consequence, we obtain an upper bound $\exp(\bigO{\beta})$ on the R\'enyi entanglement of formation of every order (including order zero), the distilled entanglement and the entanglement cost. 
Prior work of Kuwahara, Alhambra and Anshu~\cite{Kuwahara2021} only obtains a bound on R\'enyi entanglement of formation of order $\alpha$ ($\alpha$ strictly greater than $0$), albeit they achieve a better dependence on the inverse temperature. This reference also places a bound on the \emph{entanglement of purification}, but we note that this quantity is not a measure of entanglement, since it can be large even for separable states~\cite{Terhal2002}.

\subsection{Efficient classical simulation algorithm} 
Finally, in \cref{Sec:samplingalgo}, we provide an efficient classical simulation algorithm that outputs the description of the MPSs given by Eq.~\eqref{eq:gibbsstatedecomp}:

\begin{theorem}[Efficient classical simulation algorithm]
\label{thm:efficientAlgorithm}
Given $0\leq \beta <\infty$, $\eps>0$, and a local 1D Hamiltonian $H$ over $\qubits$ qubits, there exists a classical algorithm which runs in time $\poly\Paren{ \qubits , 1/\eps }$ and outputs the description of an MPS $\ket{\phi}$ of bond dimension $\chi\leq \exp({\exp(c\tilde\beta)})$, such that in expectation over the randomness of the algorithm, $\Norm{ g_\beta - \expecf{}{\dyad{\phi}} }_1 \leq \eps$.  
\end{theorem}

 The idea of classically simulating 1D Gibbs states utilizing mixtures of low-entangled states goes back to the work of White~\cite{White2009}. Most prior works on global approximations to Gibbs states via MPSs or MPOs also present classical simulation algorithms. Our result establishes that this task can be achieved with a constant bond dimension independent of the system size.

Designing low-depth quantum circuits that prepare Gibbs states is a very active area of research~\cite{Bilgin2010,bakshi2024high,rouze2025efficient,bakshi2025dobrushin,chen2025quantum,Eisert2014, Molnar2015, brandao2019finite, kato2019quantum, Kuwahara2021, kuwahara2025clustering}. Given the many advanced numerical techniques to simulate these states in 1D, such quantum algorithms are not expected to display either a practical or provable advantage over their classical counterparts in such a setting. Nevertheless, efficient state-preparation algorithms for 1D Gibbs states might be useful for basic-physics studies or as subroutines of larger quantum algorithms. Our result can be combined with existing MPS preparation schemes to achieve such a task. For instance, the results of Sch\"on, Solano, Verstraete, Cirac, and Wolf~\cite{Schon2005} together with \cref{thm:efficientAlgorithm} imply that one can efficiently construct a linear-depth circuit that, in expectation, approximates the Gibbs state within trace distance $\varepsilon$.  

Furthermore, in practice, generic MPSs satisfy additional conditions, such as normality (short-range correlated).
Under such generic conditions, there are more modern logarithmic-depth \cite{malz_preparation_2024} and  even adaptive constant-depth \cite{Smith2024} MPS preparation algorithms. By leveraging those techniques, our results could yield algorithms that are on par with or better than the state-of-the-art result based on adiabatic simulation, proposed recently by Bergamaschi and Chen~\cite{bergamaschi2025quantum}. Preparing the MPSs in our decomposition is strictly more capable since it already provides the explicit full classical description,  whereas obtaining such a description using the algorithm of~\cite{bergamaschi2025quantum} could involve prohibitively expensive tomography.
Furthermore, the randomness in our algorithm is entirely classical, allowing one to prepare the same pure state multiple times, while the procedure in \cite{bergamaschi2025quantum} requires post-selection for repeated preparation of the same sample. Such repeatability is useful, for example, to access higher statistical moments. We remark that the subsequent work~\cite{PilatowskyCameo2026} shows that 1D Gibbs states can be prepared exactly by ensembles of constant-depth circuits. 

\ifbool{journal}{\section{Proof outline}}{\section{Technical overview}}
 Our main technical tool to prove Theorem~\ref{th:GibbsStatesareMixturesOfMPS} is a general decomposition of the Gibbs state, which we call the {\it entanglement bulk decomposition}.
 We show that the Gibbs state on the full chain can be expressed as a product of a local operator $M$, the Gibbs state on the chain with the first qubit removed, and a correction $X$ which is a {\it quasilocal perturbation of the identity} (defined below).
\paragraph{Entanglement bulk decomposition.} For an $\qubits$-qubit spin chain Hamiltonian $H$ with locality $\locality$, and any $0< \gamma <1$, 
the unnormalized Gibbs state can be decomposed as
\begin{equation}
  e^{-\beta {H}}=M  \cdot e^{-\beta (H-H_1)} \cdot X. \label{eq:gibbsstatedecomposition}
\end{equation}
The local operator  $M$ is supported on the first $m = \exp(\max\{2\beta,1\}c_\gamma)$ sites, with  $c_\gamma=\locality(5+\log(1/\gamma))$, and $X$ is a quasilocal perturbation of the identity with decay $\gamma$, meaning that it can be expanded as a series
\begin{equation*}
  X=I+\sum_{\ell=1}^\qubits F_\ell,
\end{equation*}
where each $F_\ell$ is supported on the first $\ell$ sites and has an exponentially decaying operator norm $\norm{F_\ell}\leq \gamma^\ell$. The operator $H - H_1$ is the Hamiltonian with the terms acting on the first qubit removed. Equation~\eqref{eq:gibbsstatedecomposition} can be represented using the following tensor diagram:
\begin{equation}
\begin{tikzpicture}[scale=\drawingsscale, baseline={(0,-\MathAxis pt)}, thick]
		 \foreach \n in {-3,...,3}\draw {({\n},-\shorttensorlegheight) -- ({\n},\shorttensorlegheight)};
        \drawgibbs{(0,0)}{3.5}{ \footnotesize $e^{-\beta{H}}$};
    \end{tikzpicture}\cdots = 
    \begin{tikzpicture}[scale=\drawingsscale, baseline={(0,-\MathAxis pt)}, thick]
     \foreach \n in {-3,...,3}\draw {({\n},-\shorttensorstaking - \shorttensorlegheight) -- ({\n},\shorttensorstaking + \shorttensorlegheight)};
        \drawMbox{(-1,\shorttensorstaking)}{\small $M$};
        \drawgibbs{(0.5,0)}{3}{\footnotesize $e^{-\beta({H}-H_1)}$};
        \drawXbox{(0,-\shorttensorstaking)}{3.5}{\small $X$};
    \end{tikzpicture}\cdots.
    \label{eq:cuttinglemmaTens}
\end{equation}
This decomposition allows us to distill the quantum correlations involving the first qubit from those that are merely classical. As we explain below, $X$ is associated with a separable state, meaning that the local operator $M$ mediates all the entanglement between the first qubit and the rest of the chain. We provide a proof sketch of this decomposition in \cref{sec:proofoflemmaslice} below, and a full proof in \cref{sec:fullproofofmaintheorem}.

\subsection{Mixture of matrix product states}
Here, we describe how to derive \cref{th:GibbsStatesareMixturesOfMPS} from the entanglement bulk decomposition described above.

\paragraph{Ladder of local operators.} Iterating the entanglement bulk decomposition qubit by qubit, we obtain a ladder of $\tilde\qubits=\qubits-m+1$ local operators $M_j$ (each supported between qubits $j$ and $j+m-1$),
\vspace{-0.5em}
\begin{equation}
        \mathcal{M}\coloneqq M_1M_2M_3\cdots M_{\tilde\qubits}= \begin{tikzpicture}[scale=\drawingsscale, baseline={(0,-\MathAxis pt)}, thick]
      \foreach \n in {1,...,7}{\draw ({\n-4},-\shorttensorstaking-\shorttensorheight -0.2) -- ({\n-4},\shorttensorstaking + \shorttensorlegheight);}
        \drawMbox{(-1,\shorttensorstaking)}{\small $M_1$};
        \drawMbox{(0,0)}{\small $M_2$};
        \drawMbox{(1,-\shorttensorstaking)}{\small $M_3$};
               \draw[loosely dotted]  (1.2,-\shorttensorstaking-\shorttensorheight-0.3) -- (1.9,-\shorttensorstaking-\shorttensorheight-1.4);
    \end{tikzpicture} ,
\end{equation}
and a staircase of operators $X_j$ (each a quasilocal perturbation of the identity with decay $\gamma$, supported after qubit $j$),
\vspace{-1em}
 \begin{equation}
        \mathcal{X}=X_{\tilde\qubits} \cdots X_3X_2X_1 = \begin{tikzpicture}[scale=\drawingsscale, baseline={(0,-\MathAxis pt)}, thick]
      \foreach \n in {1,...,7}{\draw ({\n-4},-\shorttensorstaking - \shorttensorlegheight ) -- ({\n-4},\shorttensorstaking+\shorttensorheight+0.2);}
        \drawXbox{(1,\shorttensorstaking)}{2.5}{\small $X_3$};
        \drawXbox{(0.5,0)}{3}{\small $X_2$};
        \drawXbox{(0,-\shorttensorstaking)}{3.5}{\small $X_1$};
               \draw[loosely dotted]  (1.2,\shorttensorstaking+\shorttensorheight+0.3) -- (1.9,+\shorttensorstaking+\shorttensorheight+1.4);
    \end{tikzpicture}, 
\end{equation}
such that $\exp({-{\frac{\beta}{2}} H}) =\mathcal{M}\mathcal{X}$, with $m=\exp(\tilde\beta c_\gamma)$. We use $\beta/2$ instead of $\beta$ to be able to write a symmetric expression:
\begin{equation}
\label{eq:decomofGibbsstate}
    e^{-{\beta} H} = e^{-{\frac{\beta}{2}} H} (e^{-{\frac{\beta}{2}} H})^\dagger=\mathcal{M}\sigma\mathcal{M}^\dagger,
\end{equation}
where $\sigma=\mathcal{X}\mathcal{X}^\dagger$ is an unnormalized quantum state.

\paragraph{Separability from quasilocality.} The operator $\sigma$ is a (symmetrized) staircase of quasilocal perturbations of the identity. Generalizing the argument of~\cite{bakshi2024high}, we prove that  as long as the decay $\gamma\leq \frac{1}{56}$, such operators must be separable. Specifically, we prove in \cref{lem:decay-implies-separable} that $\sigma$ can be decomposed as a nonnegative linear combination of stabilizer product states $\ket{\bm  s}=\ket{s_1}\otimes \cdots \otimes\ket{s_\qubits}$,
\begin{equation}
\label{eq:decompofmiddle}
    \sigma= \sum_{\bm  s\in \mathcal{S}^{ \qubits}} w_{\bm s} \dyad{\bm  s},
\end{equation} where $w_{\bm  s}\geq 0$ and each label $s_j\in\mathcal{S}=\{0,1,+,-,+i,-i\}$ identifies an eigenstate of one of the Pauli matrices. 

\paragraph{Thermal state as a mixture of MPSs.} We can decompose the Gibbs state as a mixture of MPSs from the decomposition  of Eq.~\eqref{eq:decomofGibbsstate} utilizing the separability of $\sigma$. 
We represent the matrices $A_j^{b_j}$ parameterizing a general MPS [Eq.~\eqref{eq:MPSdef}] as tensors with three indices, $$\begin{tikzpicture}[scale=\drawingsscale, baseline={([yshift=-2ex]current bounding box.center)}, thick]
        \middletensor{0,0}{$A_j$}
    	\draw (0,\tensorlegheight+0.5) node {\scriptsize $b_j$};
        \draw (\tensorlegheight+0.5,0) node {\scriptsize $a$};
        \draw (-\tensorlegheight-0.5,0) node {\scriptsize $c$};
    \end{tikzpicture},$$
    so the resulting MPS is given by their horizontal contraction
    \begin{equation}
 \ket{\phi}=\frac{1}{\sqrt{\alpha}} \,  \begin{tikzpicture}[scale=\drawingsscale, baseline={([yshift=-1ex]current bounding box.center)}, thick]
    	\lefttensor{0,0}{$A_1$}
        \middletensor{1+\tensorlegheight,0}{$A_2$}
       \path (1+\tensorlegheight,0) -- node[auto=false]{\ldots} (5.5+\tensorlegheight,0);
        \righttensor{5.5+\tensorlegheight,0}{$A_\qubits$}
    \end{tikzpicture}.\label{eq:MPSTensorNetwork}
\end{equation}
To construct the MPSs in the decomposition of the Gibbs state, we observe that Eqs.~\eqref{eq:decomofGibbsstate} and \eqref{eq:decompofmiddle} yield $g_\beta =\sum_{\bm s\in \mathcal{S}^{\qubits}}\frac{w_{\bm  s}}{Z_\beta} \mathcal{M} \dyad{\bm s}\mathcal{M}^\dagger$, which is a mixture of the (unnormalized) states
\newcommand{\drawstatecirc}[2]{\filldraw [fill=statecircs] #1 circle (0.4);
      \draw #1 node {\tiny #2};}
\begin{equation}
   \mathcal{M} \ket{\bm s}=
    \begin{tikzpicture}[scale=\drawingsscale, baseline={(0,-\MathAxis pt)}, thick]
     \foreach \n in {1,...,3}{\draw ({\n-4},- \shorttensorstaking*\n +2*\shorttensorstaking- \shorttensorlegheight) -- ({\n-4},\shorttensorstaking + \shorttensorlegheight);
      \drawstatecirc{({\n-4-0.1},- \shorttensorstaking*\n +2*\shorttensorstaking- \shorttensorlegheight-0.3)}{${\n}$};}
      \foreach \n in {4,...,7}{\draw ({\n-4},-1*\shorttensorstaking ) -- ({\n-4},\shorttensorstaking + \shorttensorlegheight);}
        \drawMbox{(-1,\shorttensorstaking)}{\small $M_1$};
        \drawMbox{(0,0)}{\small $M_2$};
        \drawMbox{(1,-\shorttensorstaking)}{\small $M_3$};
               \draw[loosely dotted]  (1,-\shorttensorstaking-\shorttensorheight-0.3) -- (1.7,-\shorttensorstaking-\shorttensorheight-1.4);

    \end{tikzpicture}, 
\end{equation}
where $\begin{tikzpicture}[scale=\drawingsscale, baseline={(0,-\MathAxis pt)}, thick]
     \draw (0,0) -- (0,0.8);
      \drawstatecirc{(0,0)}{$a$};
    \end{tikzpicture} =\ket{s_a}$.
We can represent each of these states as an MPS by collecting all-but-one vertical legs into a composite virtual bond which we bend sideways:
\newcommand{\drawtopbend}{
\draw \foreach \n in {-2,...,1}{({\n},\shorttensorheight) to[out=90,in=0] ({\n-\shorttensorlegheight+\shorttensorheight},\shorttensorlegheight)}
        -- (-\Mboxwidth,\shorttensorlegheight)
         to[out=180,in=90] ({-\Mboxwidth-\marginout},\marginout) to[out=-90,in=0] ({-\Mboxwidth-2*\marginout},0);
         \draw ({2},0) -- ({2},\shorttensorlegheight);
}\newcommand{\drawbottombend}{
        
		\draw \foreach \n in {-2,...,1}{({-\n},-\shorttensorheight) to[out=-90,in=180] ({-\n+\shorttensorlegheight-\shorttensorheight},-\shorttensorlegheight)}
        -- (\Mboxwidth,-\shorttensorlegheight)
         to[out=0,in=-90] ({\Mboxwidth+\marginout},-\marginout) to[out=90,in=180] ({\Mboxwidth+2*\marginout},0);
         \draw ({-2},0) -- ({-2},-\shorttensorlegheight);
         }
\begin{equation}
    \begin{tikzpicture}[scale=\drawingsscale, baseline={(0,-\MathAxis pt)}, thick]
        \middletensor{2,0}{$A_{a}$}
    \end{tikzpicture} \,=\,    \begin{tikzpicture}[scale=\drawingsscale, baseline={(0,-\MathAxis pt)}, thick]
    \drawtopbend
    \drawbottombend
         \drawstatecirc{({-2},-1.4)}{$\tilde a$};
                \drawMbox{(0,0)}{\small $M_{\tilde a}$};
    \end{tikzpicture},
\end{equation}
for $m<a<\qubits$
(where $\tilde a=a-m+1$), with special care  taken at the boundaries:  
\begin{align*}
        \begin{tikzpicture}[scale=\drawingsscale, baseline={(0,-\MathAxis pt)}, thick]
            	\lefttensor{0,0}{$A_1$}
       \path (0,0) -- node[auto=false]{\ldots} (4.5,0);
        \middletensor{4.5,0}{$A_{m}$}
    \end{tikzpicture} =    \begin{tikzpicture}[scale=\drawingsscale, baseline={(0,-\MathAxis pt)}, thick]
     \foreach \n in {-2,...,2}{\draw ({\n},0) -- ({\n},\shorttensorlegheight);}
    \drawbottombend
         \draw ({-2},0) -- ({-2},-1.0);
         \drawstatecirc{({-2},-1.4)}{${1}$};
        \drawMbox{(0,0)}{\small $M_{1}$};
    \end{tikzpicture},&& &&
    \begin{tikzpicture}[scale=\drawingsscale, baseline={(0,-\MathAxis pt)}, thick]
        \righttensor{2,0}{$A_{ \qubits}$}
    \end{tikzpicture} =    \begin{tikzpicture}[scale=\drawingsscale, baseline={(0,-\MathAxis pt)}, thick]
		   \drawtopbend
            \foreach \n in {-1,...,1}{\draw ({\n},0) -- ({\n},-\tensorlegheight);
        \drawstatecirc{({\n},-1.4)}{}}
         \draw ({2},\shorttensorheight) -- ({2},-1.4);
         \drawstatecirc{({2},-1.4)}{$\qubits$};
         \draw ({-2},\shorttensorheight) -- ({-2},-1.4);
        \drawstatecirc{({-2},-1.4)}{};
        \drawMbox{(0,0)}{\small $M_{{\tilde \qubits}}$};
    \end{tikzpicture}.
\end{align*}

  The resulting MPS $\ket{\phi_{\bm s}}$ given by Eq.~\eqref{eq:MPSTensorNetwork} has bond dimension $\chi=2^{m-1}$, and we push its normalization 
  \begin{equation}
\label{eq:normfactor}
    \alpha_{\bm s}= \begin{tikzpicture}[scale=\drawingsscale, baseline={([yshift=-0.5ex]current bounding box.center)}, thick]
    	\lefttensor{0,0}{$A_1$}
        \middletensor{1+\tensorlegheight,0}{$A_2$}
       \path (1+\tensorlegheight,0) -- node[auto=false]{\ldots} (5.5+\tensorlegheight,0);
        \righttensor{5.5+\tensorlegheight,0}{$A_\qubits$}
    	\lefttensordag{0,1+\tensorlegheight}{$A_1^\dagger$}
        \middletensordag{1+\tensorlegheight,1+\tensorlegheight}{$A_2^\dagger$}
       \path (1+\tensorlegheight,1+\tensorlegheight) -- node[auto=false]{\ldots} (5.5+\tensorlegheight,1+\tensorlegheight);

        \righttensordag{5.5+\tensorlegheight,1+\tensorlegheight}{$A_\qubits^\dagger$}
    \end{tikzpicture},
\end{equation}
 into the associated probability  \begin{equation}
     p_{\bm  s}=\frac{\alpha_{\bm  s} w_{\bm  s}}{Z_\beta}, \label{eq:probabilities}
 \end{equation}
 so that we obtain the convex combination of MPSs claimed in 
 Eq.~\eqref{eq:mixtureMPS}. 

\subsection{Proof sketch of the entanglement bulk decomposition}
\label{sec:proofoflemmaslice}
Next, we outline the main ideas in the proof of the entanglement bulk decomposition (see \cref{sec:fullproofofmaintheorem} for a complete proof). To show the desired decomposition, $e^{-\beta H}=M \cdot e^{-\beta (H-H_1)} \cdot X$, we begin by considering the Taylor expansion of the {\it Araki expansional}~\cite{Araki1969}:
\begin{equation}
\label{eq:dysonseries}
    e^{- \beta H} \cdot e^{\beta(H -H_1)}  = \sum_{t = 0}^{\infty} { \frac{\beta^t }{t!} }  \; \sum_{\ell=0}^{\qubits}  E_{t,\ell}\,, 
\end{equation}
where $E_{t,\ell}$ collects all products $H_{j_1}\cdots H_{j_t}$ whose factors together act on exactly the qubits $1$ to $\ell$. We collect all local terms which have support $\ell \leq m$ into the operator $M$, and call the rest $R$, so that $e^{- \beta H} \cdot e^{\beta(H -H_1)}=M+R$. By defining 
\begin{equation}
\label{eq:Xdefinitionoverview}
    X=e^{\beta(H -H_1)} \cdot (I+M^{-1} R) \cdot e^{-\beta(H -H_1)}\,,
\end{equation}
we immediately recover Eq.~\eqref{eq:gibbsstatedecomposition}:
\begin{align*}
    e^{-\beta H } &= e^{-\beta H } \cdot  e^{\beta(H -H_1)} \cdot  e^{-\beta(H -H_1)} \ifbool{journal}{\\ &}{}= M \cdot (I + M^{-1} R) \cdot e^{-\beta(H -H_1)} \ifbool{journal}{\\ &}{}= M \cdot e^{-\beta(H -H_1)} \cdot X.
\end{align*}
The nontrivial part of the argument is showing that $X$ is a quasilocal perturbation of the identity,  to which we  devote the rest of this section. For this, we will need to analyze the decay of the coefficients in the Araki expansional, as well as the effects of the imaginary time evolution in Eq.~\eqref{eq:Xdefinitionoverview}.

\paragraph{Araki expansional.} 
Understanding the behavior of the Araki expansional is central to understanding properties of the Gibbs state and dates back to Araki's foundational result on the lack of phase transitions in one dimension~\cite{Araki1969}. This object also appears prominently in the proof of the separability of high-temperature Gibbs states by Bakshi, Liu, Moitra and Tang~\cite{bakshi2024high}. This reference shows that for all $\beta < \beta_c$, where $\beta_c$ only depends on the locality and degree of the Hamiltonian interactions, the Araki expansional is a quasilocal perturbation of the identity, by noting that it satisfies a clean recurrence relation:
\begin{equation}
\label{intro:recurrence}
    e^{- \beta H} \cdot e^{\beta(H -H_1)} = \sum_{t= 0}^{\infty} \frac{\beta^t}{t!} f_{t}(H, H_1) \ifbool{journal}{}{\quad \textrm{ such that }  f_t = -[H, f_{t-1}(H, H_1)] - f_{t-1}(H,H_1) H_1 \,, } 
\end{equation}
\ifbool{journal}{such that  $$f_t = -[H, f_{t-1}(H, H_1)] - f_{t-1}(H,H_1) H_1 \,,$$ }{}where $[A,B]= AB - BA$ is the standard matrix commutator.  However, the series diverges for high-dimensional systems and large $\beta$ (or low temperature). The series always converges, for any finite $\beta$, as long as the underlying Hamiltonian is one dimensional~\cite{Araki1969}. Observe, unrolling the recurrence above, an equivalent way to generate the terms is to either apply the commutator with $H$, which picks up a local term in $H$ that overlaps with the current support (an edge in the hypergraph), or multiply on the right with $-H_1$ which does not change the support. It therefore suffices to track the support/locality as we recursively generate the terms of the series. 

\paragraph{Locality in uniform hypergraphs.} We show that to track the locality of the terms in the recurrence in Eq.~\eqref{intro:recurrence}, we can strip away the quantum operators and analyze the combinatorics of paths on a uniform hypergraph that satisfy the same recurrence relation. Consider a $\locality$\nobreakdash-uniform hypergraph on the line, where the vertices are labeled $[\qubits]$ and each edge $e_i = [i, i + \locality -1]$ is an interval of length $\locality$. We introduce the definition of a \emph{valid growth set}:
\begin{definition}[Valid growth set (informal, see \cref{def:growth-sets})]
Given a $\locality$-uniform hypergraph on the line, a \emph{valid growth set} of size $t$ is a sequence of edges $(e_1, e_2, \ldots , e_t)$, such that the edge $e_1$ contains vertex $1$, and for each $2\leq s\leq t$, the edge $e_s$ intersects the support of the edges $e_1, e_2, \ldots , e_{s-1}$. 
\end{definition}

Observe, a valid growth set can increase the support by only adding edges that intersect the support constructed up to that point, which coincides with the commutator in \cref{intro:recurrence} being nonzero, and also allows for adding edges that do not change the support which accounts for right-multiplying by $-H_1$. We can therefore restrict our attention to a purely combinatorial object, the number of valid growth sets that contain $t$ edges and achieve a support $[\ell]$, which we denote by $P_{t,\ell}$, and focus on obtaining non-asymptotic bounds on it. 

\paragraph{Bound on path sums.} We can show  that $\sum_\ell P_{t,\ell}\leq \Paren{ \bigO{t/\log(t)} }^t$ by tracking the right end point of the valid growth set $(e_1, \ldots, e_s)$ after $s$ steps, which we denote by $r_s$. We define $\Delta_s = r_s - r_{s-1}$ to track the increment of $r_s$. If the growth set achieves a support $[\ell]$ after $t$ steps, then  $\sum_{s=2}^t \Delta_s = \ell - \locality$. We count the number of sequences with exactly $z$ advances, i.e. $\Delta_s \neq 0$ for exactly $z$ steps. To upper bound this quantity it suffices to first pick the $z$ locations where $\Delta_s \neq 0$, and there are at most $\binom{t}{z}$ ways to do so. Further, for each nonzero step, we have $\Delta_s \in [1, \locality-1]$, since the support cannot increase by more than $\locality$ by adding a single edge, and there are $\locality^z$ ways of picking the nonzero steps. For the remaining $t-z-1$ steps $\Delta_s =0$, and we have at most $(\locality z+1)^{t-z}$ many choices since at most $\locality z+1$ edges lie inside the support after $z$ increments. Therefore, 
\begin{align*}
    \sum_{\ell} P_{t,\ell}  &\leq \sum_{z=0}^{t} \binom{t}{z} \cdot \Paren{ \locality }^z \cdot \Paren{\locality z+1}^{t-z} \ifbool{journal}{\\ &}{}\leq (t+1) \cdot  \locality^{t}  \Paren{ \max_{z} \Paren{\frac{et}{z} }^z \cdot (z+1)^{t-z} }\,,
\end{align*}
where the second inequality follows from Stirling's approximation. To get a quick bound on the maximum, consider the log-concave function $f(x) = (et/x)^x \cdot (x+1)^{t-x}$, where $x$ takes values over the reals. The function $\log f$ is maximized between $t/(2\log(t))$ and $2t/\log(t)$, and therefore one can upper bound $\sum_{\ell} P_{t,\ell}$ by $\Paren{ \bigO{t/\log(t)} }^t$.

In \cref{thm:combinatorics}, we further show a similar bound for re-weighted path sums. In particular, we prove that for any $C\geq 1$,   \begin{equation}
\label{eq:boundoniteratedsupportintro} 
    \sum_{\ell=1}^{\qubits} C^\ell \cdot P_{t,\ell} \leq t!\left(\frac{\bigO{\locality}}{\log{(\bigO{t/C^{\locality}})}}\right)^{t}
\end{equation}  
for any $t\geq \bigO{C^{\locality}}$. We do this by first reducing to the case $\locality = 2$, where $P_{t,\ell}$ correspond to Stirling numbers of the second kind and known bounds on their weighted sums are available~\cite{Acu2024}.

\paragraph{Coefficients of the Araki expansional.} 
Equipped with our bound on the number of valid growth sets, we return to analyzing the coefficients of the Araki expansional in Eq.~\eqref{eq:dysonseries} and show that they are bounded by $\Norm{E_{t,\ell}} \leq 2^t P_{t,\ell}.$
Therefore, for a fixed $t\geq 2$ and  $\beta\geq 1$ the coefficient in the series can be bounded as follows:
\begin{equation*}
    {\frac{\beta^t}{t!}} \sum_{\ell=0}^{\qubits} \Norm{E_{t,\ell}}  \leq   {\frac{\beta^t}{t!}} \cdot 2^t\sum_{\ell=0}^{\qubits}    P_{t,\ell} \leq \Paren{  \frac{\bigO{\beta   \locality} }{\log(t)} }^t,
\end{equation*}
where we use our combinatorial bound with $C=1$. Observe, the coefficients in this series increase rapidly for small values of $t$, until $\log t$ dominates $\bigO{\beta \locality}$, and after that point, they decay at least geometrically as a function of $t$. In particular, the series starts to behave like that of~\cite{bakshi2024high} once $t \ge \exp(\bigO{\beta \locality})$, as the tail decays geometrically. Further, it is fairly straightforward to show that the operator norm of the series is bounded:
\begin{equation}
\label{eqn:op-norm-bound-intro}
    \Norm{ e^{-\beta H} \cdot e^{\beta\paren{H - H_1} } } \leq \exp\exp\Paren{ \bigO{\beta \locality} } \eqqcolon  \Lambda\,,
\end{equation}
by treating the contribution of the large terms and small terms separately. 

\paragraph{A bounded support and decaying tails decomposition.}
We split the series of the Araki expansional into \begin{align*}
    M = \sum_{t=0}^{\infty} {\frac{\beta^t}{t!}}  \sum_{\ell=0}^{m} E_{t,\ell}&&\text{and}&& R = \sum_{t = 0}^{\infty} {\frac{\beta^t}{t!}}  \sum_{\ell = m+1}^{\qubits}   E_{t,\ell}.
\end{align*} Since $M$ is a truncation, the operator norm bound from \cref{eqn:op-norm-bound-intro} continues to hold, and one can similarly bound $\norm{M^{-1}}\leq \Lambda$ (see \Cref{lem:truncating-dyson}). We can then rewrite the Araki expansional as $e^{-\beta H} \cdot e^{\beta\paren{H - H_1} } = M (I + E)$ for 
\begin{equation*}
    E\coloneqq M^{-1}R= \sum_{t = m/\locality }^{\infty}  \Paren{\frac{\beta^t}{t!}} \sum_{\ell = m+1}^{\qubits} E_{t,\ell}', 
\end{equation*}
 where $E'_{t,\ell}\coloneqq M^{-1}E_{t,\ell}$ satisfies $\norm{E'_{t,\ell}} \leq \Lambda \cdot  2^t P_{t,\ell}$. Our choice of $m$ guarantees that $I+E$ is a quasilocal perturbation of the identity. We now explain how this implies that $X$ in Eq.~\eqref{eq:Xdefinitionoverview} also satisfies this property.

\paragraph{Imaginary-time evolution of the tail.} Naively, one might expect the imaginary-time evolution in  Eq.~\eqref{eq:Xdefinitionoverview} to blow up the coefficients of the first few terms in $E$ given that the operator norm is large.
We show that this is not the case and yet again bound the coefficients by the number of valid growth sets. Recalling the Hadamard formula, 
\begin{equation*}
    e^{\beta H} \cdot Y \cdot  e^{-\beta H} = \sum_{t=0}^{\infty} \frac{\beta^t}{t!} \; [H, Y]_t\,, 
\end{equation*}
where $[A,B]_t = [A, [ A, \ldots [A, B]\ldots]]$ is the $t$-th nested commutator, we can decompose $X=I+ \sum_{\ell = m+1}^{\qubits} F_\ell$, where
\begin{align*}
     F_\ell= \sum_{t \geq  m/\locality} \; \sum_{\tau = 0}^{\infty} \; \frac{\beta^{t+\tau}}{t! \; \tau!} \; E''_{t,\tau,\ell}\,.
\end{align*}
Here, we expand each nested commutator $[H-H_1,E'_{t,\lambda}]_\tau$ into terms $[H_{i_\tau},\ldots,[H_{i_1},E'_{t,\lambda}]\ldots]$, each supported in $[\lambda]\cup e_{i_1}\cup\dots\cup e_{i_\tau}$, where $e_i$ is the interval supporting $H_i$, and $E''_{t,\tau,\ell}$ collects the terms for which this union is $[\ell]$.
To show that $X$ is a quasilocal perturbation of the identity with decay $\gamma$, we need to show that each $F_\ell$ is supported only on the first $\ell$ sites, and $\norm{F_\ell}\leq \gamma^\ell$. 
The bound on locality holds by construction.

For the bound on the norm, observe that concatenating a valid growth set with $t$ edges and support $[\lambda]$ with the edges $e_{i_1},\ldots,e_{i_\tau}$ yields a valid growth set with $t+\tau$ edges and support $[\ell]$, and every such valid growth set is obtained at most once in this way. Since $\norm{E'_{t,\lambda}}\leq\Lambda 2^t P_{t,\lambda}$, we obtain $\norm{E''_{t,\tau,\ell}}\leq\Lambda 2^{t+\tau}P_{t+\tau,\ell}$ (see \cref{lem:nested-commutator-locality-and-norm} for a complete proof). Substituting $u=t+\tau$ and using Eq.~\eqref{eq:boundoniteratedsupportintro} with $C=\gamma^{-1}$, we then obtain
\begin{align*}
    \sum_{\ell=m+1}^{\qubits}\gamma^{-\ell}\norm{F_\ell}\ifbool{journal}{&}{}
    \leq\Lambda\sum_{u\geq m/\locality}\frac{\bigO{\beta}^u}{u!}\sum_{\ell=1}^{\qubits}\gamma^{-\ell}P_{u,\ell}\ifbool{journal}{\\&}{}
    \leq\sum_{u\geq m/\locality}\Paren{\frac{\bigO{\beta\locality}}{\log(m\bigO{\gamma}^{\locality})}}^u.
\end{align*}
Our choice of $m$ guarantees that $ \log m\geq \bigO{\beta\locality}+\log(\bigO{\gamma^{-\locality}})$, so 
the ratio inside the parentheses stays below a constant small enough so that the full geometric series is bounded by $1$. In particular, this gives the desired bound $\norm{F_\ell}\leq\gamma^\ell$. See \cref{lem:img-time-evo-residual} for details.

\section{Discussion}

Our results provide an exact decomposition for the Gibbs state of any local 1D qubit Hamiltonian in terms of matrix product states, placing a rigorous and universal bound on the amount of entanglement present in 1D, as measured by the Schmidt number. Our work leaves several interesting directions open, which we enumerate below.

 \paragraph{Tightness of the bond dimension.} 
 Our bound on the bond dimension scales doubly exponentially in $\beta$. Is this scaling improvable in general, or are there explicit 1D models that asymptotically saturate it?

\paragraph{Range of entanglement and Gibbs state preparation.} What is the range of entanglement in 1D thermal states? While our results place a strong bound on the {\it amount} of thermal entanglement, they do not place direct bounds on its range. This question goes beyond studying correlations in the Gibbs state, which may also be classical in nature. In this direction, works by Kuwahara and Saito \cite{Kuwahara2022} and Kuwahara \cite{kuwahara2025clustering} have placed bounds on quantum correlations and the entanglement of formation which decay exponentially with distance, and a recent work by Scalet \cite{Scalet2026} has shown a sudden death of long-range entanglement. Building on the techniques developed here, the subsequent work~\cite{PilatowskyCameo2026} showed that the Gibbs state is in fact a mixture of products of MPSs supported on contiguous blocks whose size is independent of the system size. This implies that thermal entanglement has a strictly finite range, and that 1D Gibbs states can be prepared exactly by ensembles of constant-depth quantum circuits.

\paragraph{Quantum thermalization.} A question perhaps seemingly unrelated to the present context is that of quantum thermalization. In a nutshell, quantum thermalization asks whether a pure, low-complexity initial state evolving unitarily under a local, time-independent Hamiltonian eventually becomes locally indistinguishable from the Gibbs state at some temperature. Pilatowsky-Cameo and Choi~\cite{PilatowskyCameo2025} proved that any pure state sampled from an ensemble of low-complexity states whose average is the Gibbs state should thermalize under generic translation-invariant Hamiltonians. Since the ensemble of Ref.~\cite{PilatowskyCameo2026} consists of states preparable by constant-depth circuits, and therefore of explicitly bounded circuit complexity, it implies, at any temperature, the typical thermalization of such states under generic local translation-invariant 1D Hamiltonians~\cite{PilatowskyCameo2026}.

\paragraph{Quasilocal interactions.} Our results assume strictly local interactions, which raises the question of whether they extend to Hamiltonians with interactions of unbounded range whose strength decays sufficiently fast with distance. At sufficiently high temperatures, Putterman, Zlokapa and Cotler~\cite{Putterman2026} recently showed that the Gibbs state is exactly separable even for Hamiltonians whose terms decay exponentially with the size of their support. We expect that our results also extend to 1D Hamiltonians with sufficiently fast-decaying interaction tails, and determining the required decay is an interesting direction.

\paragraph{Higher dimensions.} In one-dimensional systems, thermal fluctuations veil quantum entanglement, only allowing a finite amount to propagate along the chain. In higher-dimensional settings, a more nontrivial interplay between locality and thermal fluctuations may occur. Quantifying entanglement in higher-dimensional systems, even in two dimensions, remains a significant challenge. Previous works~\cite{Hastings2006,Molnar2015} have approximated Gibbs states using projected entangled pair states (PEPS) or operators~\cite{Verstraete2004}, but these approaches typically require a bond dimension that scales with system size. We believe that our methods could yield exact decompositions of the Gibbs state as mixtures of PEPS with constant bond dimension within specific temperature ranges, since above a certain threshold, thermal states are known to be separable~\cite{bakshi2024high}, and we anticipate that a constant bond dimension suffices to represent states at moderately lower temperatures. However, as the temperature decreases further, phase transitions may induce long-range correlations.  Whether such states remain representable as mixtures of tensor network states is nontrivial, and answering this question would help elucidate whether these correlations are fundamentally classical or quantum in nature.

%% file: prelims.tex
\section{Preliminaries}
\label{sec:prelims}
\subsection{Basic notations and setting}
 Our physical setting is a one-dimensional, open chain of $\qubits$ qubits. Mathematically, this means that for every integer $j\in[\qubits]\coloneq\{1,2,\dots,\qubits\}$, we consider a local Hilbert space $\mathcal{H}_j=\mathbb{C}^2$. The global Hilbert space is the tensor product $\mathcal{H}=\bigotimes_{j=1}^\qubits \mathcal{H}_j$ of dimension $2^\qubits$. We say that an operator $O$  (a linear function acting on $\mathcal{H}$) is supported on a region $A\subseteq [\qubits]$ if there is another operator $O_A$ acting on $\mathcal{H}=\bigotimes_{ a\in A} \mathcal{H}_{ a}$, such that $O=O_A\otimes \bigotimes_{ j\not \in A} I_j$, where $I_j$ is the identity operator on $\mathcal{H}_j$. Henceforth, we will adopt the standard abuse of notation and not distinguish between $O_A$ and~$O$.  For a matrix $A$, we use $A^\dagger$ to denote its conjugate transpose and $\norm{A}$ to denote its operator norm. 

\begin{definition}[Pauli matrices, Pauli strings, and stabilizer product states] \label{def:paulis}
    The single-qubit Pauli operators consist of the identity $I$ and the traceless matrices $\sigma_x, \sigma_y, \sigma_z$:
\begin{equation}
    \sigma_x = \begin{pmatrix} 0 & 1 \\ 1 & 0 \end{pmatrix}, \quad
    \sigma_y = \begin{pmatrix} 0 & -i \\ i & 0 \end{pmatrix}, \quad
    \sigma_z = \begin{pmatrix} 1 & 0 \\ 0 & -1 \end{pmatrix}.
\end{equation}
These matrices are Hermitian, unitary, and involutory (square to identity). We define \textit{Pauli strings} as tensor products of the form $P = P_1 \otimes \dots \otimes P_n$, where each $P_k \in \{I, \sigma_x, \sigma_y, \sigma_z\}$. The set of all Pauli strings forms an orthogonal basis for the vector space of $2^\qubits \times 2^\qubits$  matrices. The eigenstates of these operators are denoted as $\ket{0}, \ket{1}$ for $\sigma_z$; $\ket{\pm} = (\ket{0} \pm \ket{1})/{\sqrt{2}}$ for $\sigma_x$; and $\ket{\pm i} = (\ket{0} \pm i\ket{1})/{\sqrt{2}}$ for $\sigma_y$. We define a \textit{stabilizer product state} as any tensor product $\ket{\psi} = \bigotimes_{j=1}^n \ket{s_j}$ where each $\ket{s_j}$ is one of these six eigenstates.
\end{definition}

\paragraph{Locality and quasilocality.} For an operator $O$ its \emph{support} $\supp(O) \subseteq [\qubits]$ is the subset of qubits that $O$ acts nontrivially on, i.e., the smallest set such that $O = O_{\supp(O)} \otimes I_{[n] \setminus \supp(O)}$. We say an operator is {\it local} if its support is contained within a connected region of constant size. Furthermore, we study operators whose locality is not strictly bounded, but decays exponentially, quantified as follows.

\begin{definition}[Quasilocal operators and quasilocal perturbations of the identity]
An operator $Q$ is {\it quasilocal with decay $\gamma<1$} if it can be expanded as a series
\begin{equation}
    Q=\sum_{j=1}^\qubits F_j,
\end{equation}
where $F_j$ is supported between sites $1$ and $j$ and has an exponentially decaying operator norm $\norm{F_j}\leq \gamma^j$. Further, for such $Q$, we say that $X=I+Q$ is a {\it quasilocal perturbation of the identity with decay $\gamma$}.
\end{definition}

\paragraph{Hamiltonian and Gibbs state.}
A \emph{Hamiltonian} on $\qubits$ qubits is represented by any Hermitian operator $H$.  We study a $\locality$-local Hamiltonian on a line, which is a sum of \emph{local terms} $H_1, \dots , H_{n-\locality+1}$, with $H = \sum_{j = 1}^{n-\locality+1} H_j$ and each operator $H_j$ supported on sites $\{j, j+1, \dots , j + \locality - 1 \}$.  For normalization, we assume that $\norm{H_j} \leq 1$. For such a Hamiltonian, we denote its {\it Gibbs state} by $g_\beta=\exp(-\beta H)/Z_\beta$, where $\beta\geq 0$ is the {\it inverse temperature} and $Z_\beta=\tr(\exp(-\beta H))$ is the {\it partition function}.

\paragraph{Matrix product state.} A {\it matrix product state with bond dimension $\bonddim$} is a state represented by a set of matrices  $ A_1^{b_1}A_2^{b_2}\cdots A_\qubits^{b_\qubits}$, where $b_j\in\{0,1\}$ and $j\in[\qubits]$. For $j\neq1,\qubits$,  $A_j^{b_j}$ has dimension $\bonddim\times \bonddim$, while $A_1^{b_1}$ and $A_\qubits^{b_\qubits}$ are $1\times \bonddim$ and $\bonddim \times 1$-dimensional, respectively. The associated MPS is given by $$\ket{\phi}=\frac{1}{\sqrt{\alpha}}\sum_{b_j\in\{0,1\}} A_1^{b_1}A_2^{b_2}\cdots A_\qubits^{b_\qubits} \ket{b_1 b_2\cdots  b_\qubits},$$
where $\ket{b}$ is the eigenstate of Pauli $\sigma_z$ and the number $\alpha>0$ ensures normalization $\langle\phi|\phi\rangle =1$.

\subsection{Entanglement in mixed states}
Here, we introduce quantum information theoretical measures of quantum entanglement in mixed states. For a comprehensive review, see \cite{Horodecki2009}. Recall the notion of separability:
\begin{definition}[Separable operator]
    We say that an operator $O$ is {\it separable} if it can be written as a positive linear combination of product states,    
    \begin{equation}
        O=\sum_i w_i \ket{\psi_i}\bra{\psi_i},
    \end{equation}
    i.e., $\ket{\psi_i}=\ket{\psi_i^{(1)}}\otimes\cdots \otimes \ket{\psi_i^{(\qubits)}}$, and $w_i>0$.
    We say a state is separable if its density matrix $\rho$ is separable.
\end{definition}

When a state is separable, it possesses no quantum entanglement. Measures of entanglement serve to quantify how much a state deviates from separability. For a density matrix $\sigma$ we represent by $\mathcal{U}(\sigma)$ the set of all its {\it unravelings} $U$, which is a set of tuples $(p_\psi,\ket{\psi})$ such that $\sum_{(p_\psi,\ket{\psi})\in U} p_\psi \dyad{\psi}=\sigma$, where $p_\psi\geq 0$. We denote a {\it bipartition} of the chain by $A:B$, where $A,B\subseteq [n]$ are disjoint and $A\cup B=[n]$.
\begin{definition}[R\'enyi entanglement of formation] 
    Given a density matrix $\sigma$ acting on a bipartition $A:B$, its {\it R\'enyi entanglement of formation of order $\alpha\geq 0$} is given by
    \begin{equation}
        E_{f,\alpha}(\sigma_{AB})=\min_{U\in \mathcal{U}(\sigma)} \sum_{p_\psi \in U} p_\psi S_\alpha(\psi_A),
    \end{equation}
    where $S_\alpha(\psi_A)=\frac{1}{1-\alpha}\log(\tr(\psi_A^\alpha))$ is the R{\'e}nyi entropy of $\psi_A=\tr_B(\dyad{\psi})$ ($\alpha\neq 1$), with  $S_1(\rho)=-\tr(\rho \log\rho)$ being the usual von Neumann entropy.
\end{definition}
For $\alpha=1$, the entanglement of formation~\cite{Bennett1996B} quantifies how many Bell pairs are required, on average, to prepare the mixed state. It upper bounds the so-called entanglement of distillation, which measures the number of Bell pairs that can be extracted from the state. In this work, we bound a more stringent measure of entanglement, defined as follows.
\begin{definition}[Schmidt number]
    Given a density matrix $\sigma$ acting on a bipartition $A:B$, its {\it Schmidt number} is
    \begin{equation}
        \mathrm{SN}(\sigma_{AB})=\min_{U\in \mathcal{U}(\sigma)} \max_{p_\psi \in U} \mathrm{rank}(\psi_A),
    \end{equation}
    where the rank of $\psi_A$ corresponds to the Schmidt rank of the pure state $\ket{\psi}$.
\end{definition}

\begin{fact}
    For any $\alpha \geq 0$, any mixed state $\sigma$, and any bipartition $A:B$,
    \begin{equation}
       E_{f,\alpha}(\sigma_{AB})\leq \log(\mathrm{SN}(\sigma_{AB})).
    \end{equation}
\end{fact}
The Schmidt number is a strictly more stringent measure than the R\'enyi entanglement of formation of any order. The entanglement of formation minimizes an average, while the Schmidt number minimizes the worst case. We will utilize matrix product states to control the Schmidt number. Specifically, one can show the following.
\begin{observation}
Let $\rho=\sum_i p_i \dyad{\psi_i}$ be a mixture of matrix product states of bond dimension at most $\bonddim$. Then, for any contiguous bipartition, i.e., $A=\{1,2,\dots,\ell\}$ and $B=\{\ell+1,\ell+2,\dots \qubits \}$, the Schmidt number of $\rho$ is upper bounded as $\mathrm{SN}(\rho)\leq \bonddim.$ More generally, for an arbitrary bipartition of the chain, $\mathrm{SN}(\rho)\leq \bonddim^M,$
where $M$ is the number of boundaries generated by the bipartition.
\end{observation}

%% file: separability.tex
\section{Separability from quasilocality}

Here, we state and prove several sufficient conditions for separability that can be inferred from \cite{bakshi2024high}. We leverage these results for our decomposition of the thermal state in one-dimensional systems; however, we believe some of them may be of independent interest.
\subsection{Separability of quasilocal perturbations of the identity}
It is well known that a perturbation of the identity $X=I+\Delta$ is separable as long as $\Delta$ is sufficiently small. Gurvits and Barnum \cite{Gurvits2002,Gurvits2003sep} proved that this general statement is only true if $\Delta$ has a norm which is exponentially small in the system size. However, it was recognized in  \cite{bakshi2024high} that if one additionally requires that $\Delta$ is quasilocal, then this result can be greatly improved:

\begin{proposition}[Quasilocal perturbations of the identity are separable]
\label{prop:quasi-local-perturbations}
Let $X$ be Hermitian and such that
\begin{equation*}
    X =I+\sum_{\ell=1}^{\qubits} F_\ell  , \quad \textrm{ where }  \abs{\supp(F_\ell)}  \leq \ell   \; \textrm{ and } \Norm{F_\ell} \leq \gamma^{\ell} \,
\end{equation*}
for $\gamma\leq \frac{1}{8}$.
 Then $X$ can be decomposed into a nonnegative combination of stabilizer product states $\dyad{s_1}\otimes \dyad{s_2} \otimes \ldots \otimes \dyad{s_n}$, and in particular $X$ is separable.
\end{proposition}

\begin{proof}
Consider a distribution over $[\qubits]$ that draws $\ell \in [\qubits]$ with probability $p_{\ell} = \frac{2^{-\ell}}{1-2^{-\qubits}}.$ This is a valid probability distribution since $\sum_{\ell =1}^{\qubits} p_\ell = 1$. 
Draw an $\ell \in [\qubits]$ from this probability distribution and output the operator $Z_\ell = I + F_{\ell}/p_{\ell}$. Observe $Z_\ell$ is an unbiased estimator of $X$, i.e., $\expecf{}{Z_{\ell}} = I + \sum_{\ell=1}^{\qubits} F_\ell  = X$. Further,  
\begin{equation*}
    \frac{\Norm{F_\ell}}{p_\ell} \leq  (1-2^{-\qubits})  \Paren{2 {\gamma} }^\ell \leq 4^{-\ell}\,.
\end{equation*}
Further note that $F_\ell$ can be taken to be Hermitian (otherwise, since $X$ is Hermitian, we can replace $F_\ell$ with $(F_\ell+F_\ell^\dagger)/2$ in the decomposition of $X$).
We prove in \cref{lem:sep-of-perturbations with bounded-support} below that such $Z_\ell$ must be a nonnegative linear combination of stabilizer product states.  Then $X$ is also a nonnegative linear combination of stabilizer product states. 
\end{proof}

\begin{lemma}[Local perturbations of the identity are separable]
\label{lem:sep-of-perturbations with bounded-support}
Let $X = I + Y$ where $Y$ is Hermitian, $|\supp(Y)| \leq \ell$ for an integer $0\leq \ell\leq \qubits$, and $\Norm{Y} \leq 4^{-\ell}$. Then, $X$ can be decomposed into a nonnegative combination of stabilizer product states.
\end{lemma}
\begin{proof}
Let $\calS$ be a set of $\ell$ sites containing $\supp(Y)$ and let $\calP_{\calS}$ be the set of Pauli matrices supported on $\calS$ and $I$ otherwise. Then, we can expand $Y$ in the Pauli-string basis $Y = \sum_{P \in \calP_{\calS} } \alpha_P  P$,  where $\alpha_P = \frac{1}{2^{\qubits}} \tr\Paren{ Y P }$ is real since $Y$ is Hermitian. Observe, for any Pauli string $P$, its trace norm is $\Norm{P}_{\tr} = 2^\qubits$, and thus    $\abs{\alpha_{P}} \leq \frac{1}{2^\qubits}  \cdot \Norm{Y} \cdot \Norm{P}_{\tr} \leq \Norm{Y}\leq 4^{-\ell}.$ Next, observe that $\abs{\calP_\calS} =4^\ell$, and therefore 
\begin{equation*}
    I + Y = \frac{1}{4^\ell}\bigg( \sum_{P \in \calP_\calS } I + c_P  P  \bigg),
\end{equation*}
where $c_P=4^\ell \cdot \alpha_P$ satisfies $|c_P| \leq 1$.

We only need to show that $I + c P$ is a nonnegative linear combination of stabilizer product states for any $|c| \leq 1$. Since $P$ is a tensor product of Pauli matrices supported on $\mathcal{S}$, it can be diagonalized in a basis of stabilizer product states $\ket{\bm s}_{\mathcal{S}} = \bigotimes_{j \in \mathcal{S}} \ket{s_j}$, with corresponding eigenvalues $\lambda_{\bm s} \in \{ \pm 1\}$. We can therefore write the spectral decomposition directly as:
\begin{equation*}
I + c P = \sum_{\bm s} (1 + c \lambda_{\bm s}) \dyad{\bm s}_{\mathcal{S}} \otimes I_{\overline{\mathcal{S}}} .
\end{equation*}
Since $|c| \leq 1$, the coefficients $w_{\bm s} = 1 + c \lambda_{\bm s}$ are nonnegative. Furthermore, the identity on the complementary subsystem $\overline\calS$ can be trivially expanded as a sum of stabilizer product states.
\end{proof}

\subsection{Separability of staircases of quasilocal perturbations of the identity}
So far, we have shown that a single Hermitian quasilocal perturbation of the identity is separable. Even more generally, we show that the staggered product (a {\it staircase}) of quasilocal perturbations of the identity, not necessarily Hermitian themselves, but appropriately symmetrized, is also separable.
\begin{theorem}[Staircases of quasilocal perturbations of the identity are separable]\label{lem:decay-implies-separable}
Given an $\qubits$-qubit system, let $X_1, X_2, \ldots, X_\qubits$ be a sequence of quasilocal perturbations of the identity with decay $\gamma\leq 1/56$, such that  $\supp(X_i) \subseteq [i, \qubits]$. That is, each $X_i$ admits the following decomposition:
\begin{align*}
    X_i = I + \sum_{\ell=1}^{\qubits-i+1} F^{(i)}_\ell\,, \ifbool{journal}{}{\quad \textrm{ where } \; \supp(F^{(i)}_\ell) \subseteq [i, i + \ell -1]   \; \textrm{ and } \|F^{(i)}_\ell\| \leq \gamma^{\ell}\,.}
\end{align*}
\ifbool{journal}{ where $\supp(F^{(i)}_\ell) \subseteq [i, i + \ell -1],$  and $\|F^{(i)}_\ell\| \leq \gamma^{\ell}$. }{} Then,  $\Paren{ X_\qubits \cdot X_{\qubits-1} \cdots X_1 } \Paren{ X_\qubits \cdot X_{\qubits-1} \cdots X_1 }^\dagger$ can be written as a nonnegative linear combination of stabilizer product states.
\end{theorem}
\noindent We prove this result via a sequence of lemmas. 

\begin{lemma}
\label{lemm:pinpauli}
    Let $c\in [-1,1]$ and $P=P_1\otimes P_2\otimes \cdots\otimes P_n$ be a Pauli string. Then
    \begin{equation}
    \label{lemm:plusminuspauli}
        (I+cP)=\sum_{\pm}  \dyad{s_{\pm}}\otimes (I+ c_\pm P_{[2,\qubits]}),
    \end{equation}
    where $\ket{s_\pm}$ are the orthogonal eigenstates of $P_1$, $\abs{c_\pm}=\abs{c}$ and $P_{[2,\qubits]}=P_2\otimes P_3\otimes \cdots P_\qubits$.
\end{lemma}
\begin{proof}
   If $P_1=I$ the statement is trivially true taking $\ket{s_\pm}$ to be any pair of orthogonal stabilizer qubit states, and $c_\pm=c$. If $P_1$ is a nonidentity Pauli, then observe that
     \begin{equation}
        (I+cP)=\frac{1}{2}(I+P_1)\otimes (I+c P_{[2,\qubits]}) +\frac{1}{2}(I-P_1)\otimes (I-c P_{[2,\qubits]}).
    \end{equation}
    Further, $\frac{1}{2}(I\pm P_1)=\dyad{s_\pm}$ where $\ket{s_\pm}$ is the $\pm$ eigenstate of $P_1$.
\end{proof}
\ifbool{journal}{\begin{widetext}}{}
\begin{lemma}
\label{lemm:symmetrizedproductPAPpisseparable}
Let $P,P'$ and $A$ be Pauli strings with support contained in $[\ell]$, $[\ell']$ and $[a]$, respectively, where $a\in\{0,1,\dots,\qubits\}$ and $[0]=\emptyset$. Assume that $z,z'$ are complex numbers with $|z| \leq \delta^\ell$, $|z'| \leq \delta^{\ell'}$ and $c\in [-\delta^a,\delta^a]$, for $0<\delta\leq1/7$. Then 
\begin{equation}
   \tfrac{1}{2}[(I+zP)(I+cA)(I+z'P')^\dagger+(I+z'P')(I+cA)(I+zP)^\dagger] =\sum_{j=1}^N w_j \dyad{s_j} \otimes (I+c_jB_j),
\end{equation}
where $N$ is some integer, $B_j$ is a Pauli string supported in $[2,b_j]$, weighted by $c_j$ with $\abs{c_j}\leq\delta^{b_j-1}$, and the states $\ket{s_j}\in\{\ket{0},\ket{1},\ket{+},\ket{-},\ket{+i},\ket{-i}\}$ are single-qubit stabilizer states, with weights $w_j\geq 0$.
\end{lemma}
\begin{proof}
If $a=0$, the claim follows from the case $c=0$ by multiplying all weights by $1+c\geq0$. Hence, assume $a\geq1$.
    
Using the general expansion
    \begin{align}
        (I+Y)(I+B)&(I+Y')^\dagger+(I+Y')(I+B)(I+Y)^\dagger =2(I+B) + (Y^\dagger +Y)+ (Y'^\dagger +Y')\\&+ (Y B +B Y^\dagger)+(Y' B +B Y'^\dagger) + (Y' Y^\dagger +Y Y'^\dagger)+ (Y'B Y^\dagger +YB Y'^\dagger),
    \end{align}
    we can write
        \begin{align}\frac{1}{2}\left [(I+zP)(I+cA)(I+z'P')^\dagger+(I+z'P')(I+cA)(I+zP)^\dagger\right ] =I +\sum_{j=1}^7 O_j,\label{eq:lemm_full_expression_IpOj}
        \end{align}
        where 
        \begin{align*}
            O_1&=cA,& O_2&=\mathrm{Re}(z)P,&  O_3&=\mathrm{Re}(z')P',\\
O_4&= c(z P A +z^*A P)/2,& O_5&=c(z'P' A +z'^*A P')/2,& O_6&=(z'^* z P P' +z'z^* P' P)/2,\\  O_7&= c(z'^* z PA P' +z'z^*P'A P)/2.
\end{align*}
Since Pauli strings always commute or anticommute, we have $O_i=x_i P_i$ for a Pauli string  $P_i$ and real $x_i$. For example, $x_1=c$ and $P_1=A$, while $x_7=c\mathrm{Re}(z'^* z)$ if $P'AP=PAP'=P_7$ and $x_7=c\mathrm{Im}(z'^* z)$ if $P'AP=-PAP'=iP_7$. 

Let $w_i=1/7$ for $i\in \{1,\dots,7\}$ and $Z_i=I+\tilde{c}_i P_i$, with $\tilde{c}_i=x_i/w_i$, so we have $\eqref{eq:lemm_full_expression_IpOj}=\sum_{i=1}^7 w_i Z_i$. We only need to show that  $|\tilde{c}_i|\leq \delta^{s_i-1}$ with $\supp(P_i)\subseteq [s_i]$. We analyze the cases $i=1$ and $i=7$ (the other cases are analogous). For $i=1$, note that $|\tilde{c}_1|=7|c|\leq 7 \delta^a\leq \delta^{a-1}$ (where we used $\delta\leq 1/7$), and $\supp(P_1)=\supp(A)\subseteq [a]$. For $i=7$, note that $|x_7|\leq |c||z||z'|$, and hence $|\tilde{c}_7|\leq 7 |c||z||z'|\leq \delta^{a+\ell+\ell'-1}\leq  \delta^{\max(a,\ell,\ell')-1}$, while $\supp(P_7)= \supp(P'AP)\subseteq \supp(P')\cup\supp(A)\cup\supp(P)\subseteq [\ell']\cup[a]\cup[\ell]=[\max(a,\ell,\ell')]$. 

Finally, applying \cref{lemm:pinpauli} to each $Z_i$ gives $Z_i=\sum_{\pm}\dyad{s_{i,\pm}}\otimes(I+c_{i,\pm}B_i)$, where $B_i$ is the restriction of $P_i$ to $[2,\qubits]$, which is supported in $[2,s_i]$, and $|c_{i,\pm}|=|\tilde c_i|\leq\delta^{s_i-1}$. This proves the claim with $N=14$, indexing $j=(i,\pm)$, with $w_j=1/7$, $\ket{s_j}=\ket{s_{i,\pm}}$, $c_j=c_{i,\pm}$, $B_j=B_i$ and $b_j=s_i$.
\end{proof}

\begin{lemma}
\label{lem:separabilityconjbyquasilocalperturbations}
Let  $A$ be a Pauli string supported in $[a]$, $c\in [-\delta^a,\delta^a]$, with $\delta=1/7$ and $X$ a quasilocal perturbation of the identity with decay $\gamma\leq 1/56$. Then
\begin{equation}
    X (I+cA)X^\dagger =\sum_{j} w_j \dyad{s_j} \otimes (I+c_jB_j),
\end{equation}
with $B_j$, $c_j$, $w_j$, and $\ket{s_j}$ as in \cref{lemm:symmetrizedproductPAPpisseparable}.
\end{lemma}
\begin{proof}
We expand $X=I +\sum_{\ell=1}^\qubits F_\ell$ with $\norm{F_\ell}\leq \gamma^\ell$ and $\supp(F_\ell)\subseteq [1,\ell]$.
    For each $\ell\in[\qubits]$ we define
    \begin{equation*}
        p_\ell=\frac{2^{-\ell}}{1-2^{-\qubits}},
        \qquad Z_\ell=I+\frac{F_\ell}{p_\ell}.
    \end{equation*}
    Then $\sum_{\ell=1}^{\qubits}p_\ell=1$ and $\sum_{\ell=1}^{\qubits}p_\ell Z_\ell=X$. Further, we expand $F_\ell=\sum_P\alpha_P^{(\ell)}P$ over all $4^\ell$ Pauli strings on $[\ell]$, including the identity (here $\alpha_P^{(\ell)}$ can be complex in general). Taking the trace on these $\ell$ sites, $|\alpha_P^{(\ell)}|=|2^{-\ell}\tr(PF_\ell)|\leq\norm{F_\ell}\leq\gamma^\ell$. We define
    \begin{equation*}
        z_P^{(\ell)}=\frac{4^\ell}{p_\ell}\alpha_P^{(\ell)}=(1-2^{-\qubits})8^\ell\alpha_P^{(\ell)},
        \qquad
        Z_\ell=\frac{1}{4^\ell}\sum_P(I+z_P^{(\ell)}P),
    \end{equation*}
    so that $|z_P^{(\ell)}|\leq(8\gamma)^\ell\leq(\tfrac{1}{7})^\ell$. Then, using the fact that the left-hand side is Hermitian,
    \begin{align}
        X(I+cA)X^\dagger
        &=\sum_{\ell,\ell'=1}^{\qubits}
        \frac{p_\ell p_{\ell'}}{2\cdot 4^{\ell+\ell'}}\sum_{P,P'}
        \Big[(I+z_P^{(\ell)}P)(I+cA)(I+z_{P'}^{(\ell')}P')^\dagger\nonumber\\
        &\hspace{10em}+(I+z_{P'}^{(\ell')}P')(I+cA)(I+z_P^{(\ell)}P)^\dagger\Big],
        \label{eq:firstdecompXicAX}
    \end{align}
    where $P$ and $P'$ run over all Pauli strings on $[\ell]$ and $[\ell']$, respectively. We obtain the desired result by applying \cref{lemm:symmetrizedproductPAPpisseparable} to each symmetrized pair.
\end{proof}

\begin{proof}[Proof of \cref{lem:decay-implies-separable}]
Simply iterate \cref{lem:separabilityconjbyquasilocalperturbations}. Start with $X=X_1$ and $c=0, A=I$, yielding that $X_1X_1^\dagger =\sum_j w_j \dyad{s_j} \otimes (I+c_jB_j).$ This {\it pins} the first site. Then
\begin{equation}
    X_2X_1X_1^\dagger X_2^\dagger =\sum_j w_j \dyad{s_j} \otimes X_2(I+c_jB_j)X_2^\dagger,
\end{equation}
since $X_2$ has support only after the second site. Now, since   $|c_j|\leq 1/7^{b}$, with $\supp(B_j)\subseteq [2,b+1]$, we can again apply \cref{lem:separabilityconjbyquasilocalperturbations} to each term $ X_2(I+c_jB_j)X_2^\dagger$, relabeling the $i$th qubit to $i-1$ and ignoring the first site. This {\it pins} the second site,
\begin{equation}
    X_2X_1X_1^\dagger X_2^\dagger =\sum_{j,j'} w_j w'_{j,j'} \dyad{s_j} \otimes \dyad{s'_{j,j'}}\otimes (I+c'_{j,j'}B'_{j,j'}),
\end{equation}
and we may keep recursing to obtain the required decomposition. Here it is worth noting that the primed variables depend both on $j$ and $j'$: this is capturing the potential classical correlations in the separable operator.
\end{proof}
\ifbool{journal}{%
    \end{widetext}
}{}

%% file: combinatorics.tex
\section{Valid growth sets on the line}
\label{sec:combinatorics}

We introduce the notion of a {\it valid growth set} on the uniform hypergraph on a line. Non-asymptotic estimates of the number of such sets are crucial to our proof, as they will directly bound the norm of the higher order terms in the Araki expansional and under imaginary time evolution.

We consider an $n$-vertex $\locality$-uniform hypergraph on the line, meaning the vertex set is $[n]$ and the edges consist of \(\locality\) consecutive vertices, i.e., $\calE=\{ \{i,i+1,\dots,i+\locality-1\}:\ i\in [n-\locality+1]\}$. On such a hypergraph, we define \emph{valid growth sets}:
\begin{definition}[Valid growth set]
\label{def:growth-sets}
A \emph{valid growth set} of size $t\geq 1$ is a sequence of edges $e_1, \dots , e_t\in \calE$ with the following properties:
\begin{itemize}
\item {\it Starting vertex.} The hyperedge $e_1$ contains vertex~$1$.
\item {\it Connected component.} For all $2\leq s \leq t$, $e_s$ intersects the union of $e_1, \dots , e_{s-1}$.
\end{itemize}   
We define $P_{t,\ell}^{\smash{(\locality)}}$ to be the number of distinct valid growth sets that consist of $t$ edges and such that their union is equal to $[\ell]$.   We sometimes omit the $\locality$ in the superscript and just write $P_{t,\ell}$ when it is clear from context.

\end{definition}

\DeclareRobustCommand{\stirling}{\genfrac\{\}{0pt}{}}

For $\locality=2$, the numbers $P_{t,\ell}^{\smash{(2)}}$ are simply the Stirling numbers of the second kind $\stirling{t}{\ell-1}$, which count the number of ways to partition a set of $t$ objects into $\ell-1$ nonempty sets. In our setting, the objects are the $t$ time-steps, and the $\ell-1$ nonempty sets correspond to the $\ell-1$ distinct graph edges that must be covered to form the connected component of length $\ell$. For higher $\locality>2$, there is no such combinatorial interpretation, but nevertheless the numbers $P^{\smash{(\locality)}}_{t,\ell}$ can be bounded in terms of $P^{\smash{(2)}}_{t,\ell}$.
\begin{lemma}[Coarse graining in $\locality$]
\label{lemm:coarse-graining}
    For any $\locality\geq 2$ and $j\in \{0,1,\dots,\locality-2\}$,
    $$P_{t,(\locality-1)\ell-j}^{(\locality)}\leq (\locality-1)^{t-1} P_{t,\ell}^{(2)}.$$
\end{lemma}
\begin{proof}
    Given a line with $\qubits$ vertices, $\{1,2,\dots,\qubits\}$, and a $\locality$-uniform hypergraph connecting every set of $\locality$ consecutive vertices, we consider a {\it coarse graining} $\{B_1,B_2,\dots,B_{n/(\locality-1)}\}$, obtained by binning every $\locality-1$ sites, $B_s=\{(s-1)(\locality -1)+1,(s-1)(\locality -1)+2,\dots,s(\locality -1) \}$. We define a graph on this coarse-grained line where each edge connects only nearest neighbors $\{B_s,B_{s+1}\}$ and is identified with the set $\widetilde{e}$ containing the $\locality-1$ edges in the original hypergraph that start in $B_s$.  One can then map a valid growth set $({e}_1,\dots,{e}_t)$ supported on $[(\locality-1)\ell-j]$ ($j\in\{0,\dots,\locality-2\}$) in the original hypergraph to a valid growth set $(\widetilde{e}_1,\dots, \widetilde{e}_t)$ on the coarse-grained chain with support $\{B_1,\dots B_\ell\}$. This mapping has degeneracy $(\locality-1)^{t-1}$, corresponding to the $(\locality-1)$ choices of $e_{s}\in \widetilde{e}_s$ for each $s\in \{2,3,\dots,t\}$ (not counting the first edge which is uniquely set by the valid-growth-set starting-vertex condition). Hence, we obtain the stated bound. 
\end{proof}
We show that if the uniform hypergraph is formed by taking adjacent vertices on the line, the number of distinct valid growth sets cannot grow too quickly. More generally, we provide a bound on the weighted sums of the number of valid growth sets.

\begin{theorem}[Bound on weighted partial sums of numbers of valid growth sets]
\label{thm:combinatorics}
Given a $\locality$-uniform hypergraph over the line, any  $C\geq 1$, and any $\eta\geq 2$, the weighted partial sum of all valid growth sets of size $t\geq \eta C^{\locality-1}$ is bounded as follows:
\begin{equation}
\label{eq:boundoniteratedsupport} 
    \sum_{\ell=0}^{\qubits} C^\ell \cdot P_{t,\ell} \leq t!\left(\frac{\alpha_\eta {(\locality-1)}}{\log{(t/C^{\locality-1} +\frac{1}{2})}}\right)^{t},
\end{equation}
where $\alpha_\eta\coloneqq e^{1/\log\eta}\log(\eta+\frac{1}{2})/(\log\eta-\log\log\eta)$. 

\end{theorem}

\begin{proof}

    We prove the case of $\locality=2$ first, which will imply the general case by Lemma~\ref{lemm:coarse-graining}. We assume $t\geq \eta C$. Since  $P^{(2)}_{t,\ell}=\stirling{t}{\ell-1}$ are the Stirling numbers of the second kind, the desired weighted sum can be expressed in terms of the so-called Touchard polynomials, 
    defined by $T_t(x)\coloneqq \sum_{\ell}{\stirling{t}{\ell}} x^\ell$.
Specifically,
    \begin{equation*}
        \sum_{\ell=0}^{\qubits} C^{\ell}P_{t,\ell}^{(2)}\leq CT_t(C).
\end{equation*}

Acu, Adell, and Rasa~\cite{Acu2024} prove the following bound for any $t>C$:
\begin{equation*}
    T_t(C)\leq \frac{t^t}{e^{t+C}}\exp\Paren{ t \,{\Theta(t/C)}},
\end{equation*}
where $\Theta(x)=\frac{1}{\log(x)}-\log(\log(x)-\log\log(x))$. We can simplify this expression to 
\begin{align*}
    \sum_{\ell=0}^{\qubits} C^{\ell}P^{(2)}_{t,\ell}&\leq C\frac{t^t}{e^{t+C}}\exp( t\,\Theta(t/C))
    \\ & \leq  t!\exp( t\,\Theta(t/C))
    \\ & \leq  t!\exp\Paren{t\log\Paren{ {\frac{\alpha_\eta}{\log(t/C+1/2)}}}}
    \\&= t! \Paren{\frac{\alpha_\eta}{\log(t/C+1/2)}}^t,
\end{align*}
where we used that $\Theta(x) \leq \log(\alpha_\eta/\log(x+1/2))$ for $x\geq \eta$. This completes the proof for $\locality=2$.

Now for the case of $\locality\geq 3$, we simply apply Lemma~\ref{lemm:coarse-graining}. For $t\geq \eta C^{\locality-1}$,
\begin{align*}
    \sum_{\ell=0}^{\qubits} C^\ell \cdot P^{(\locality)}_{t,\ell}&=\sum_{\ell=0}^{\lceil\qubits/(\locality-1)\rceil} \sum_{j=0}^{\locality-2} C^{(\locality-1)\ell-j}\cdot P^{(\locality)}_{t,(\locality-1)\ell-j}\\&\leq \sum_{\ell=0}^{\lceil\qubits/(\locality-1)\rceil} C^{(\locality-1)\ell}\sum_{j=0}^{\locality-2} P^{(\locality)}_{t,(\locality-1)\ell-j}\\&\leq (\locality-1)^t  \sum_{\ell=0}^{\lceil\qubits/(\locality-1)\rceil} (C^{\locality-1})^{\ell} \cdot P^{(2)}_{t,\ell} \ifbool{journal}{\\&}{}\leq t!\left(\frac{\alpha_\eta {(\locality-1)}}{\log{(t/C^{\locality-1}+\frac{1}{2})}}\right)^{t}.\qedhere
\end{align*}
\end{proof}

%% file: entanglement-bulk-decomp.tex
\section{Entanglement bulk decomposition}

\label{sec:fullproofofmaintheorem}
In this section, we prove the entanglement bulk decomposition, which is the main technical tool behind our main result.
\begin{theorem}[Entanglement bulk decomposition]
\label{th:entanglement-bulk-decomposition}
 For any geometrically $\locality$-local Hamiltonian on a spin chain of length $\qubits$ and any $0< \gamma <1$, 
the unnormalized Gibbs state at any inverse temperature $0 \leq \beta<\infty$ can be decomposed as
\begin{equation}
  e^{-\beta {H}}=M  \cdot e^{-\beta (H-H_1)} \cdot X, \label{eq:gibbsstatedecomposition2}
\end{equation}
where  $M$ is supported on the first $m = \lceil\exp(\max\{2\beta,1\} c_\gamma)\rceil$ sites, with  $c_\gamma=\locality(5+\log(1/\gamma))$, and $X$ is a quasilocal perturbation of the identity with decay $\gamma$.
\end{theorem}

We can derive \cref{th:GibbsStatesareMixturesOfMPS} by iterating the entanglement bulk decomposition of \cref{th:entanglement-bulk-decomposition} with $\beta\to\beta/2$ a total of $\tilde{\qubits}=\qubits -m+1$ times (at the last step, we instead take $M_{\tilde{\qubits}}=e^{-\beta H_{[\tilde{\qubits},\qubits]}/2}$ and $X_{\tilde{\qubits}}=I$, where $H_{[\tilde{\qubits},\qubits]}=\sum_{j=\tilde{\qubits}}^{\qubits-\locality+1}H_j$), yielding
\begin{equation*}
    e^{-\beta H} =   \Paren{ M_1 \cdot M_2 \cdots M_{\tilde{\qubits}} } \cdot \sigma \cdot \Paren{ M_1 \cdot M_2 \cdots M_{\tilde{\qubits}} }^{\dagger},
\end{equation*}
with $\sigma=\Paren{ X_{\tilde{\qubits}} \cdot X_{\tilde{\qubits}-1} \cdots X_1 }    \cdot \Paren{ X_{\tilde{\qubits}} \cdot X_{\tilde{\qubits}-1} \cdots X_1 }^{\dagger}$. Picking the decay $\gamma=1/56$ [so $c_\gamma=\locality(5+\log 56)\leq 10\locality$, and $\max\{2\cdot\frac{\beta}{2},1\}=\tilde{\beta}$], each $X_i$ satisfies the requirements of \cref{lem:decay-implies-separable},  and therefore the operator $\sigma$ can be written as a nonnegative linear combination of stabilizer product states, and in particular is separable.

To prove the entanglement bulk decomposition, we begin by showing that the Araki expansional can be decomposed into operators with bounded support whose locality can be related to the number of valid growth sets introduced in \cref{sec:combinatorics}. We will henceforth assume that $\beta\geq 1/2$ (otherwise all inequalities can be seen to hold replacing $\beta\in [0,1/2)$ with $1/2$).

\begin{lemma}[Locality of the Araki expansional]
\label{lem:dyson-series-1d}
Given $\beta \geq 1/2$ and a $\locality$-local Hamiltonian $H$, let $H_1$ be the terms that contain the first site. Then, the Araki expansional admits the following decomposition:
\begin{equation*}
    e^{-\beta H} \cdot e^{\beta\paren{H - H_1} } = I + \sum_{t = 1}^{\infty} { \frac{\beta^t }{t!} }  \; \sum_{\ell=0}^{\qubits}  E_{t,\ell} \,, 
\end{equation*}
where each $E_{t,\ell}$ is supported on $[\ell]$ and $\Norm{E_{t,\ell}} \leq 2^{t}\cdot P^{(\locality)}_{t,\ell}$, where $P^{(\locality)}_{t,\ell}$ is the number of distinct valid growth sets with $t$ edges, and support $[\ell]$.
\end{lemma}

\begin{proof}
Taylor expanding each of the exponentials, we have
\begin{align}
    e^{-\beta H}e^{\beta (H - H_1)} & =  \sum_{ s = 0}^{\infty} \frac{\beta^s}{s!} (-H)^s   \sum_{\tau= 0}^{\infty} \frac{\beta^\tau}{\tau! } \Paren{H - H_{1}}^{\tau}  \nonumber \\
    & = \sum_{t = 0}^{\infty} \frac{\beta^t}{t!}  \underbrace{ \sum_{\tau = 0}^t  \binom{t}{\tau} \Paren{ -H}^{\tau}\Paren{H - H_1}^{t-\tau } }_{f_t\Paren{H, H_1} } \,. 
\label{eqn:taylor-expand-dyson}
\end{align}
Next, we observe that  $f_0\Paren{H,H_1} = I$ and  $f_t\Paren{H, H_1}$ satisfies the following recurrence relation: 
\begin{align*}
    f_t\Paren{H, H_1} & = - H \cdot f_{t-1}\Paren{H, H_1} + f_{t-1}\Paren{H, H_1} \cdot \Paren{H- H_1} \\
    & = - [ H,  f_{t-1}\Paren{H, H_1} ] - f_{t-1}\Paren{H, H_1} H_1 \\
    & = - \sum_{j\geq 2}[ H_j,  f_{t-1}\Paren{H, H_1} ] - H_1 f_{t-1}\Paren{H, H_1}.
\end{align*}
Unrolling this recurrence, each term of $f_t$ corresponds to a sequence of edges $(e_{j_1},\dots,e_{j_t})$: at step $s$ we either take the commutator with $H_{j_s}$, $j_s\geq 2$, or left multiply by $-H_1$, $j_s=1$. Since $f_0=I$, we must have $j_1=1$, and a commutator vanishes unless $e_{j_s}$ intersects the previous edges, so the nonzero terms are indexed by valid growth sets. Each such term is supported on the union of its edges and has norm at most $2^{t-1}$, since each commutator at most doubles the norm. Therefore, $f_t(H,H_1)$ admits the following decomposition:
\begin{equation*}
    f_t(H,H_1) = \sum_{ \ell = 0}^{\qubits} c_{t,\ell} F_{t,\ell}\,,
\end{equation*}
where $c_{t,\ell}F_{t,\ell}$ collects the terms whose valid growth set has union $[\ell]$, $\abs{c_{t,\ell}} \leq 2^t \cdot P_{t,\ell}$, $F_{t,\ell}$ has support that is $[\ell]$, and $\Norm{F_{t,\ell}} \leq 1$. Plugging this back into \cref{eqn:taylor-expand-dyson} completes the proof.
\end{proof}

Next, we prove a bound that we will require in our analysis.

\ifbool{journal}{\begin{widetext}}{}
\begin{lemma}[Supremum of a function]
\label{lem:suprema-of-concave-functions}
Let $\Delta \geq 3$ and consider 
\begin{equation*}
    t_* = \argmax_{t\geq 2}\; \Paren{ \frac{\Delta}{ \log(t)} }^t \,.
\end{equation*}
Then, $t_* \leq \exp\Paren{\Delta}$ and the maximum value is bounded as $\Paren{\frac{\Delta}{\log(t_*)}}^{t_*} \leq \exp\Paren{\exp\Paren{\Delta}}$. 
\end{lemma}

\begin{proof}
The function $f:[2,\infty)\to\mathbb{R}$ given by $f(x)=(\Delta/\log(x))^x$ is unimodal, meaning that it is increasing to a unique maximum and then afterwards it is strictly decreasing (it is decreasing for large enough $x$ such that $\Delta/\log(x)\leq 1$).  By differentiating, we see that its unique maximum $x_*$ satisfies $\log(\Delta) =\log\log(x_*)+\frac{1}{\log(x_*)}$. Then, we have $t_*\leq \exp(\Delta)$ and
\begin{align*}
    \Paren{\frac{\Delta}{\log(t_*)}}^{t_*} &\leq \Paren{\frac{\Delta}{\log(x_*)}}^{x_*}\leq \Paren{\frac{\Delta}{\log(x_*)}}^{\exp(\Delta)} \leq \exp(\exp(\Delta)),
\end{align*}
where in the last inequality we used that 
 $\Delta/\log(x_*)=\exp(1/\log(x_*))\leq e$, which holds because $x_*> e$. 
\end{proof}

Now we break up the Araki expansional into the part that acts on the first $m$ sites and the part that escapes beyond the first $m$ sites. We show that the latter series has rapidly decaying coefficients, since the denominator $t!$  eventually suppresses the number of valid growth sets, $\sum_{\ell} P_{t,\ell}$, which we upper bounded in \cref{thm:combinatorics}.

\begin{lemma}[Tail bounds for the Araki expansional]
\label{lem:truncating-dyson}
Given $\beta \geq 1/2$, an integer $m\geq\locality$, and a $\locality$-local Hamiltonian $H$, let $H_1$ be the terms that contain the first site, $H_{[1,m]}=\sum_{j=1}^{m-\locality+1}H_j$, and let $\Lambda =  \exp\Paren{\exp\Paren{9\,\beta \locality} } $. Then, 
\begin{equation*}
    e^{-\beta H} \cdot e^{\beta\paren{H - H_1} } = M (I + E) \,, 
\end{equation*}
where $M\coloneqq e^{-\beta H_{[1,m]}}e^{\beta(H_{[1,m]}-H_1)}$ is supported on sites $[m]$ and satisfies $\norm{M},\norm{M^{-1}}    \leq \Lambda$. Additionally, $E$ admits the following shell decomposition: 
\begin{equation*}
    E = \sum_{t = \lceil m/\locality\rceil }^{\infty}  {\frac{\beta^t}{t!}} \sum_{\ell = m+1}^{\qubits} E_{t,\ell}'\,,  \ifbool{journal}{}{\quad \textrm{ where } \Norm{E'_{t,\ell}} \leq \Lambda \cdot  2^t P_{t,\ell} \textrm{ and } \supp(E_{t,\ell}')\subseteq [\ell]\,.}
\end{equation*}
\ifbool{journal}{where  $\norm{E'_{t,\ell}} \leq \Lambda \cdot  2^t P_{t,\ell}$ and $\supp(E_{t,\ell}')\subseteq [\ell]$.}{}
\end{lemma}
\begin{proof}
First, observe that the coefficients of the series expansion of $ e^{-\beta H} \cdot e^{\beta\paren{H - H_1} } $ for a fixed $t\geq 2$ satisfy the following:
\begin{equation*}
\begin{split}
 {\frac{\beta^t}{t!}} \sum_{\ell=0}^{\qubits} \Norm{E_{t,\ell}} & \leq   {\frac{(2\beta)^t}{t!}} \sum_{\ell=0}^{\qubits}   P_{t,\ell} \leq \Paren{     \frac{8 \beta \locality }{\log(t)} }^t,  
\end{split}
\end{equation*}
where the last inequality follows from \cref{thm:combinatorics} with $C=1$ and $\eta=2$, for which $\alpha_2<4$. 
For $t=0$ the coefficient can be explicitly seen to be bounded by $1$, and for $t=1$ we have $f_1(H,H_1)=-H_1$, whose contribution to the norm is at most $\beta$, using $\norm{H_1}\leq 1$. Therefore
\begin{equation}
        \Norm{ e^{-\beta H} \cdot e^{\beta\paren{H - H_1} } }  \leq 1 + \beta +  \sum_{t=2}^{\infty} \Paren{ \frac{8 \beta \locality}{\log(t)} }^t .
\end{equation}

For any $t\geq t_* \coloneqq \lceil\exp\Paren{ 16 \beta \locality }\rceil$, we have $8\beta \locality/\log(t)\leq 1/2$, so the sum decays geometrically after $t_*$. For the terms before $t_*$, we use the bound of Lemma~\ref{lem:suprema-of-concave-functions}, with $\Delta=8 \beta \locality$. 
\begin{equation}
\label{eqn:op-norm-dyson}
    \begin{split}
        \Norm{ e^{-\beta H} \cdot e^{\beta\paren{H - H_1} } } 
        & \leq 1 +\beta +\sum_{t=2}^{t_*-1} \Paren{ \frac{8 \beta \locality}{\log(t)} }^t + \sum_{t=t_*}^{\infty} \Paren{ \frac{8 \beta \locality}{\log(t)} }^t   \\
        &\leq 2 + \beta +  (t_*-2) \exp(\exp(8\beta \locality))
        \\
        &\leq 2 +  \exp\Paren{ 16 \beta \locality } \exp(\exp(8\beta \locality)) \\
        & \leq \exp\Paren{\exp\Paren{9 \beta \locality} } \,,
    \end{split}
\end{equation}
where we used $\beta\leq\exp(\exp(8\beta\locality))$ and $t_*-1\leq\exp(16\beta\locality)$ in the third inequality, and $2+\exp(8x)\exp(\exp(4x))\leq \exp(\exp(4.5 x))$ for $x=2\beta\locality\geq 1$ in the last inequality.
We pick $M = \sum_{t=0}^{\infty} \Paren{\frac{\beta^t}{t!}}  \sum_{\ell=0}^{m} E_{t,\ell}$, i.e. the truncation of the series in \cref{lem:dyson-series-1d} to $\ell \in [m]$. It also follows from  \cref{lem:dyson-series-1d} that the support of $M$ is $[m]$
. We can then define the remainder as $R = \sum_{t = 0}^{\infty} \Paren{\frac{\beta^t}{t!}}  \sum_{\ell = m+1}^{\qubits}   E_{t,\ell}$.  Since $M$ is a truncation of the Araki expansional,  the same operator norm bound continues to hold:
\begin{equation*}
    \Norm{ M } \leq  \exp\Paren{\exp\Paren{9 \beta \locality } }.
\end{equation*}
Here, as in the proof of \cref{lem:dyson-series-1d}, $E_{t,\ell}$ collects the terms generated by sequences of interactions whose supports have union $[\ell]$. Hence $M=e^{-\beta H_{[1,m]}}e^{\beta(H_{[1,m]}-H_1)}$, where $H_{[1,m]}=\sum_{j=1}^{m-\locality+1}H_j$. So $M$ is invertible, and its inverse $M^{-1}=e^{-\beta(H_{[1,m]}-H_1)}e^{\beta H_{[1,m]}}$ has the same kind of expansion, to which the bound in \cref{eqn:op-norm-dyson} applies as well.

Now, we focus on analyzing $E = M^{-1}R$. Let $\Lambda =  \exp\Paren{\exp\Paren{9 \beta \locality} } $. Observe, for the support to reach $m$, we require at least $m/\locality$ many steps and therefore, the terms in the series are zero for $t < m/\locality$. Then, $E$ can be expanded as follows: 

\begin{equation*}
    E = M^{-1} \cdot \Paren{  \sum_{t = 0}^{\infty} {\frac{\beta^t}{t!}}  \sum_{\ell = m+1}^{\qubits}   E_{t,\ell}  } = \sum_{t = \lceil m/\locality\rceil }^{\infty}  {\frac{\beta^t}{t!}} \sum_{\ell = m+1}^{\qubits} E_{t,\ell}',   
\end{equation*}
where $E_{t,\ell}' = M^{-1} E_{t,\ell}$ acts only on $[\ell]$ since $M^{-1}$ acts only on $[m]$, and $\norm{ E_{t,\ell}' } \leq \norm{M^{-1}} \norm{E_{t,\ell}} \leq \Lambda  \cdot 2^t P_{t,\ell}$, which yields the claim. 
\end{proof}

Next, we bound the operator norm and locality of a nested commutator, defined by
$[A,B]_t = [A, [A,B]_{t-1}]$ with $[A,B]_0=B$.

\begin{lemma}[Locality of nested commutators]
\label{lem:nested-commutator-locality-and-norm}
 Suppose, for $t\geq 1$,
\begin{equation*}
\begin{gathered}
A_t=\sum_{\lambda=1}^{\qubits}A_{t,\lambda},
\qquad
\supp(A_{t,\lambda})\subseteq[\lambda],
\qquad
\Norm{A_{t,\lambda}}\leq a_tP_{t,\lambda},
\end{gathered}
\end{equation*}
where $a_t\geq 0$. Then, for every integer $\tau\geq 0$, we may decompose
\begin{equation*}
\begin{gathered}
[H,A_t]_\tau=\sum_{\ell=1}^{\qubits}B_{t,\tau,\ell},
\qquad
\supp(B_{t,\tau,\ell})\subseteq[\ell],
\qquad
\Norm{B_{t,\tau,\ell}}\leq a_t2^\tau P_{t+\tau,\ell}.
\end{gathered}
\end{equation*}
\end{lemma}
\begin{proof}
Let $e_i=\{i,i+1,\ldots,i+\locality-1\}$ contain the support of the interaction $H_i$, so $e_{i_s}$ is the edge chosen at the $s$-th additional commutator. Let us define the set $\mathcal C_\tau(\lambda,\ell)$ of indices $(i_1,\ldots,i_\tau)$ which are possible continuations from $[\lambda]$ to $[\ell]$, meaning that 
\begin{gather*}
e_{i_s}\cap\left([\lambda]\cup\bigcup_{r<s}e_{i_r}\right)\neq\emptyset\:\:\text{ and }\:\:
[\lambda]\cup\bigcup_{s=1}^{\tau}e_{i_s}=[\ell].
\end{gather*}
Then, splitting a valid growth set after its first $t$ edges gives
\begin{equation}
\label{eq:continuation-count}
P_{t+\tau,\ell}
=\sum_{\lambda=1}^{\ell}P_{t,\lambda}
\left|\mathcal C_\tau(\lambda,\ell)\right|.
\end{equation}
Define
\begin{equation*}
B_{t,\tau,\ell}
=\sum_{\lambda=1}^{\ell}
\sum_{(i_1,\ldots,i_\tau)\in\mathcal C_\tau(\lambda,\ell)}
[H_{i_\tau},\ldots,[H_{i_1},A_{t,\lambda}]\ldots],
\end{equation*}
so
\begin{align*}
[H,A_t]_\tau
&=\sum_{\lambda=1}^{\qubits}[H,A_{t,\lambda}]_\tau
=\sum_{\lambda=1}^{\qubits}\sum_{\ell=\lambda}^{\qubits}
\sum_{(i_1,\ldots,i_\tau)\in\mathcal C_\tau(\lambda,\ell)}
[H_{i_\tau},\ldots,[H_{i_1},A_{t,\lambda}]\ldots]
=\sum_{\ell=1}^{\qubits}B_{t,\tau,\ell}.
\end{align*}
Every summand is supported in $[\lambda]\cup\bigcup_{s=1}^{\tau}e_{i_s}
=[\ell]$,
so $\supp(B_{t,\tau,\ell})\subseteq[\ell]$. Using $\norm{H_i}\leq 1$ and
\cref{eq:continuation-count},
\begin{align*}
\Norm{B_{t,\tau,\ell}}
&\leq\sum_{\lambda=1}^{\ell}
\sum_{(i_1,\ldots,i_\tau)\in\mathcal C_\tau(\lambda,\ell)}
\Norm{[H_{i_\tau},\ldots,[H_{i_1},A_{t,\lambda}]\ldots]}\\
&\leq 2^\tau\sum_{\lambda=1}^{\ell}
\left|\mathcal C_\tau(\lambda,\ell)\right|\Norm{A_{t,\lambda}}
\leq a_t2^\tau\sum_{\lambda=1}^{\ell}
P_{t,\lambda}\left|\mathcal C_\tau(\lambda,\ell)\right|
=a_t2^\tau P_{t+\tau,\ell}.\qedhere
\end{align*}
\end{proof}

Next, we show that imaginary time evolution of the error term $E$ does not change the decay behavior, and continues to admit a shell decomposition as earlier.  Lemma~\ref{lem:img-time-evo-residual} below directly implies the entanglement bulk decomposition  stated in \cref{th:entanglement-bulk-decomposition}. 

\begin{lemma}[Quasilocality of imaginary time evolution of the residual]
\label{lem:img-time-evo-residual}
Given $\beta \geq 1/2$, a $\locality$-local Hamiltonian $H$ and $0<\gamma<1$, let $M$ and $E$  be as defined in \cref{lem:truncating-dyson}, with $m = \lceil\exp( 2c_\gamma\beta)\rceil $, $c_\gamma=\locality(5+\log(1/\gamma))$. Then, $X\coloneq e^{\beta (H-H_1)}(I+ E )  e^{-\beta (H-H_1)}$ is a quasilocal perturbation of the identity with decay $\gamma$. Specifically, 
\begin{equation*}
    X = I+\sum_{\ell = m+1 }^{\qubits}  F_{\ell} , \quad \textrm{ where }  \supp(F_\ell) \subseteq [\ell],   \; \textrm{ and } \Norm{F_\ell} \leq \gamma^\ell \,. 
\end{equation*}
\end{lemma}
\begin{proof}
 Recall, from \cref{lem:truncating-dyson}, that
\begin{gather*}
E=\sum_{t=\lceil m/\locality\rceil}^{\infty}\frac{\beta^t}{t!}
       \sum_{\lambda=m+1}^{\qubits}E'_{t,\lambda},
\end{gather*}
where $\supp(E'_{t,\lambda})\subseteq[\lambda],$ $\norm{E'_{t,\lambda}}\leq\Lambda\,2^tP_{t,\lambda}$, and $\Lambda =  \exp\Paren{\exp\Paren{9 \beta \locality} }$. By the Hadamard formula $e^{xB} A e^{-xB} =\sum_{t=0}^\infty \frac{x^t}{t!} [B,A]_t$, we have
\begin{equation*}
e^{\beta (H-H_1) } E e^{-\beta (H-H_1)} = \sum_{\tau = 0}^{\infty} \frac{\beta^\tau}{\tau!} \sum_{t=\lceil m/\locality\rceil}^{\infty}\frac{\beta^t}{t!}\Big[H-H_1,\sum_{\lambda=m+1}^{\qubits}E'_{t,\lambda}\Big]_\tau.
\end{equation*}
Applying \cref{lem:nested-commutator-locality-and-norm} with $a_t=\Lambda2^t$, we obtain
\begin{equation*}
\Big[H-H_1,\sum_{\lambda=m+1}^{\qubits}E'_{t,\lambda}\Big]_\tau
=\sum_{\ell=m+1}^{\qubits}E''_{t,\tau,\ell},
\qquad
\Norm{E''_{t,\tau,\ell}}\leq
\Lambda2^{t+\tau}P_{t+\tau,\ell},
\end{equation*}
where each $E''_{t,\tau,\ell}$ is supported in $[\ell]$. Define
\begin{equation*}
F_\ell=\sum_{t\geq\lceil m/\locality\rceil}\sum_{\tau\geq0}
\frac{\beta^{t+\tau}}{t!\,\tau!}E''_{t,\tau,\ell}.
\end{equation*}
We obtain the desired bound $\Norm{F_\ell} \leq \gamma^\ell$ from the following calculation: 
\begin{align*}
\sum_{\ell=m+1}^{\qubits}\gamma^{-\ell}\Norm{F_\ell}
&\leq
\Lambda\sum_{t\geq\lceil m/\locality\rceil}\sum_{\tau\geq0}
\frac{(2\beta)^{t+\tau}}{t!\,\tau!}
\sum_{\ell=1}^{\qubits}\gamma^{-\ell}P_{t+\tau,\ell}
\\
&=
\Lambda\sum_{u\geq\lceil m/\locality\rceil}
\frac{(2\beta)^u}{u!}
\sum_{t=\lceil m/\locality\rceil}^{u}\binom ut
\sum_{\ell=1}^{\qubits}\gamma^{-\ell}P_{u,\ell}
\\
&\leq
\Lambda\sum_{u\geq\lceil m/\locality\rceil}
\frac{(4\beta)^u}{u!}
\sum_{\ell=1}^{\qubits}\gamma^{-\ell}P_{u,\ell}
\\
&\leq
\sum_{u\geq\lceil m/\locality\rceil}
\Paren{
\frac{6\beta(\locality-1)\Lambda^{1/u}}
{\log(u\gamma^{\locality-1}+\tfrac12)}
}^u
\\
&\leq
\sum_{u\geq\lceil m/\locality\rceil}
\Paren{
\frac{6\beta(\locality-1)\exp(\locality e^{9\beta\locality}/m)}
{\log(m/\locality)-(\locality-1)\log(1/\gamma)}
}^u
\\
&\leq
\sum_{u\geq\lceil m/\locality\rceil}
\Paren{
\frac{6\beta(\locality-1)\exp(\locality e^{-\beta\locality})}
{10\beta\locality-\log\locality+(2\beta\locality-\locality+1)\log(1/\gamma)}
}^u
\\
&\leq
\sum_{u\geq\lceil m/\locality\rceil}
\Paren{
\frac{6}{10-2/e}\cdot\frac{\locality-1}{\locality}\exp(\locality e^{-\locality/2})
}^u
\\ &\leq
\sum_{u\geq\lceil m/\locality\rceil}\Paren{\frac{17}{20}}^{u}
=
\frac{20}{3}\Paren{\frac{17}{20}}^{\lceil m/\locality\rceil}
\leq1.
\end{align*}
The steps use, in order, the bound on $E''_{t,\tau,\ell}$; the substitution $u=t+\tau$; the binomial identity; \cref{thm:combinatorics} with $C=\gamma^{-1}$ and $\eta=e^9$, for which $\alpha_{e^9}<3/2$, where we use $u\geq m/\locality\geq e^9C^{\locality-1}$; $u\geq m/\locality$ and $\log\Lambda=e^{9\beta\locality}$; the definition of $m$; and $\beta\geq1/2$, $\log\locality\leq 2\beta\locality/e$, $\log(1/\gamma)>0$, and $\frac{\locality-1}{\locality}\exp(\locality e^{-\locality/2})\leq 1.31$ for all integers $\locality\geq 2$. The last inequality follows from $\lceil m/\locality\rceil\geq e^9\geq \log(3/20)/\log(17/20)$. 
\end{proof}

\begin{proof}[Proof of \cref{th:entanglement-bulk-decomposition}]
Let $m=\lceil\exp(\max\{2\beta,1\}c_\gamma)\rceil$. \cref{lem:truncating-dyson} gives $M$ and $E$ such that $e^{-\beta H}e^{\beta(H-H_1)}=M(I+E)$ and $M$ is supported on $[m]$. Multiplying on the right by $e^{-\beta(H-H_1)}$ gives
\begin{equation*}
  e^{-\beta H}=M\cdot e^{-\beta(H-H_1)}\cdot X,\qquad X\coloneqq e^{\beta(H-H_1)}(I+E)\,e^{-\beta(H-H_1)},
\end{equation*}
and, by \cref{lem:img-time-evo-residual}, $X$ is a quasilocal perturbation of the identity with decay $\gamma$.
\end{proof}

%% file: prep_algorithm.tex
\section{Sampling algorithm}
\label{Sec:samplingalgo}
We outline the algorithm stated in Theorem~\ref{thm:efficientAlgorithm}. Each step of the algorithm is labeled by an index $t\in \{0,1,\dots, \qubits-m+1\}$ and three parameters: a Pauli string $ B^{(t)}$ supported on $[t+1,t+b_t]$, a real number $c^{(t)}$ with $|c^{(t)}|\leq \delta^{b_t}$, and a stabilizer product state on the first $t$ qubits $\ket{s^{(1)},s^{(2)},\cdots s^{(t)}}$. 
 \subsection{Normalized distribution}

 Let us apply the (symmetrized) entanglement bulk decomposition  together with \cref{lem:separabilityconjbyquasilocalperturbations} once, i.e., we pin the first qubit. We obtain
\begin{align*}
    \exp(-\beta H)&=M_1 \cdot e^{-\beta H_{[2,\qubits]}/2} \cdot X_1X_1^\dagger \cdot e^{-\beta H_{[2,\qubits]}/2} \cdot M_1^\dagger\\&=\sum_j w_j^{(1)} \underbrace{M_1 \dyad{s^{(1)}_{j_1}}\otimes e^{-\beta H_{[2,\qubits]}/2}(I+c^{(1)}_{j_1}B_{j_1}^{(1)})e^{-\beta H_{[2,\qubits]}/2}M_1^\dagger}_{G^{(1)}_{j_1}},
\end{align*}
where $H_{[a,b]}$ denotes the Hamiltonian with all terms with support outside of $[a,b]$ removed. We want to write this expression as a proper probability mixture of quantum states, so we normalize all operators: 
 \begin{equation}
     g_\beta = \sum_{j_1} p^{(1)}_{j_1}\rho^{(1)}_{j_1},\quad\text{ where } \quad p^{(1)}_{j_1}=w^{(1)}_{j_1}\frac{\tr(G^{(1)}_{j_1})}{\tr(e^{-\beta H})}\quad\text{ and }\quad\rho^{(1)}_{j_1}=\frac{G^{(1)}_{j_1}}{\tr(G^{(1)}_{j_1})}.
 \end{equation} We move to the second qubit.  We again apply the entanglement bulk decomposition and pin the second qubit to obtain $g_\beta =\sum_{j_1, j_2} p^{(1)}_{j_1} p^{(2)}_{j_2|j_1} \rho_{j_1,j_2}$,
 where $p^{(2)}_{j_2|j_1}=w^{(2)}_{j_2|j_1}\tr(G^{(2)}_{j_1,j_2})/\tr(G^{(1)}_{j_1})$ is the  conditional probability of sampling $j_2$, given we sampled $j_1$, and $\rho^{(2)}_{j_1,j_2}=G^{(2)}_{j_1,j_2}/\tr(G^{(2)}_{j_1,j_2})$, where
 \begin{equation}
     G^{(2)}_{j_1,j_2}=M_1 M_2 \dyad{s_{j_1}^{(1)}s_{j_2}^{(2)}} \otimes e^{-\beta H_{[3,\qubits]}/2}(I+c_{j_1,j_2}^{(2)}B_{j_1,j_2}^{(2)})e^{-\beta H_{[3,\qubits]}/2}M_2^\dagger M_1^\dagger.
 \end{equation}

 More generally, after $t$ steps, we can write
 \begin{equation}
     g_\beta = \sum_{j_1,j_2,\dots ,j_t} p^{(1)}_{j_1} p^{(2)}_{j_2|j_1}\dots p^{(t)}_{j_t|j_1,\dots,j_{t-1}} \,\,\rho_{j_1,\dots,j_t},
 \end{equation}
where 
\begin{align}
\label{eq:probsweightsandstates}
    p^{(t)}_{j_t|j_1,\dots,j_{t-1}}=w^{(t)}_{j_t|j_1,\dots,j_{t-1}}\frac{\tr(G^{(t)}_{j_1,\dots,j_{t}})}{\tr(G^{(t-1)}_{j_1,\dots,j_{t-1}})} &&\text{ and }&&\rho^{(t)}_{j_1,\dots,j_t}=\frac{G^{(t)}_{j_1,\dots,j_t}}{\tr(G^{(t)}_{j_1,\dots,j_t})}.
\end{align}
 Here,
 \begin{equation}
     G^{(t)}_{j_1,\dots,j_t}=M_1 \cdots M_t\left[ \sigma_{j_1,\dots,j_t} \otimes e^{-\beta H_{[t+1,\qubits]}/2}(I+c_{j_1,\cdots,j_t}^{(t)}B_{j_1,\cdots,j_t}^{(t)})e^{-\beta H_{[t+1,\qubits]}/2}\right]M_t^\dagger\cdots  M_1^\dagger,
 \end{equation}
 where $\sigma_{j_1,\dots,j_t}=\dyad{ s_{j_1}^{(1)}\cdots s_{j_t}^{(t)}}$ and $G^{(0)}=\exp(-\beta H)$. Once we reach the end of the chain, the resulting  $\rho_{j_1,\cdots,j_n}$ is a pure matrix product state, as desired. 
 \subsection{Computing the probabilities}
 All quantities defining the probabilities $p^{\smash{(t)}}_{j_t|j_1,\dots,j_{t-1}}$ can be computed explicitly. The coefficients $c^{\smash{(t)}}_{j_1,\dots,j_t}$ are determined by the coefficients of the quasilocal terms $F_\ell$ of $X_t$ in the Pauli-string basis, $\alpha^{\smash{(\ell)}}_P=2^{-\ell}\tr(PF_\ell)$. The terms  $F_\ell$ are defined in \cref{lem:img-time-evo-residual}  through a series of nested commutators, but they admit a simpler closed form, as follows:
\begin{lemma}[Closed form of the quasilocal terms]
\label{lem:closedformF}
Let $F_\ell$ be as in \cref{lem:truncating-dyson,lem:img-time-evo-residual}. For $m<\ell\leq\qubits$, 
\begin{equation*}
    F_\ell=X^{(\ell)}-X^{(\ell-1)}\otimes I,
\end{equation*}
where $X^{(\ell)}=e^{\beta(H_{[1,\ell]}-H_1)}\,M^{-1}\,e^{-\beta H_{[1,\ell]}}$, $X^{(m)}=I$  and $H_{[1,\ell]}=\sum_{j=1}^{\ell-\locality+1}H_j$. 
\end{lemma}
\begin{proof}
Note that, by construction, $X=e^{\beta(H-H_1)} M^{-1}e^{-\beta H}=I+\sum_{\ell =m+1}^n F_{\ell}$. Now,  fix $m\leq n'\leq\qubits$ and let $H'=H_{[1,n']}$. Since $F_\ell$ collects the terms generated by sequences of interactions whose supports have union $[\ell]$, replacing $H$ by $H'$ (and keeping $M$) retains exactly the terms with $\ell\leq n'$, so that
\begin{equation*}
    X^{(n')}=e^{\beta(H'-H_1)} M^{-1}e^{-\beta H'}=I+\sum_{\ell =m+1}^{\smash{n'}} F_{\ell}.
\end{equation*}
 Taking differences between consecutive values of $n'$ gives the claim.
\end{proof}
We truncate the expansion $X_t=I+\sum_\ell F_\ell$ at $\ell_{\max}= O(\log(\qubits /\varepsilon))$, since the truncated terms have norm bounded by $\gamma^{\ell_{\max}-1}= 1/\text{poly}(\qubits/\varepsilon)$. Hence only the $F_\ell$ with $\ell\leq\ell_{\max}$ are needed, and by \cref{lem:closedformF} they can be computed by exponentiating $2^\ell\times 2^\ell$ matrices, in time $\mathrm{poly}(\qubits/\varepsilon)$.
 
 The index $j_t=(\ell,\ell',P,P',i,\pm)$ runs over $\ell,\ell'\in[\qubits-t+1]$, Pauli strings $P$ and $P'$ supported on the first $\ell$ and $\ell'$ qubits of $[t,\qubits]$, the seven terms $i\in\{1,\dots,7\}$ in Eq.~\eqref{eq:lemm_full_expression_IpOj}, and the two signs in Eq.~\eqref{lemm:plusminuspauli}, with weights
\begin{equation*}
    w^{(t)}_{j_t|j_1,\dots,j_{t-1}}=\frac{p_\ell\, p_{\ell'}}{7\cdot 4^{\ell+\ell'}}\times\begin{cases}1+c^{(t-1)}_{j_1,\dots,j_{t-1}} & \text{if } \supp(B^{(t-1)}_{j_1,\dots,j_{t-1}})=\emptyset,\\ 1 & \text{otherwise,}\end{cases}
\end{equation*}
where $p_\ell=2^{-\ell}/(1-2^{-(\qubits-t+1)})$ (see Eq.~\eqref{eq:firstdecompXicAX} and the proof of \cref{lemm:symmetrizedproductPAPpisseparable}). The operators $M_t$ are given in closed form by \cref{lem:truncating-dyson}. After the truncation, $\ell,\ell'\leq\ell_{\max}$, so the index $j_t$ takes $T=O(\exp(\ell_{\max})^2)=\text{poly}(\qubits/\varepsilon)$ values (coming from the Pauli strings $P,P'$, which are supported on at most $\ell_{\max}$ qubits). 

 Ideally, we would just need to sample $p^{\smash{(t)}}_{j_t|j_1,\dots,j_{t-1}}$ at each step $t$. However, we can only approximate $\tr(G^{\smash{(t)}}_{j_1,\dots,j_t})$, including the partition function $G^{\smash{(0)}}=\exp(-\beta H)$, up to constant multiplicative error. Thus, we leverage the weak approximate counting algorithms of Jerrum and Sinclair \cite{Sinclair1989}, which construct a random walk which may backtrack (at any step $t$ we may go back to step $t-1$), and allow us to ultimately sample from the distribution $p_{j_1,\dots,j_t}=p^{\smash{(1)}}_{j_1} p^{\smash{(2)}}_{j_2|j_1}\dots p^{\smash{(t)}}_{j_t|j_1,\dots,j_{t-1}}$ at the end of the chain $t=\tilde{\qubits}\coloneqq \qubits - m+1$ with only ever needing to estimate $p_{j_1,\dots,j_t}$ up to a multiplicative constant error for any intermediate step $t$.
\subsection{Efficiently sampling the distribution}

Theorem 5.8 of Ref.~\cite{bakshi2024high} shows that\footnote{The constants in this theorem can be straightforwardly modified to apply to our current setting.} we can sample efficiently from the {\it leaf} distribution $p_{j_1,\dots,j_{\tilde \qubits}}$ if we are able to respond to three types of queries: 
\begin{enumerate}
    \item Estimate the {\it distortion} $\tr(G^{(t)}_{j_1,\dots,j_{t}})\prod_{s=1}^{t}\sum_{j_s'}w^{(s)}_{j_s'|j_1,\dots,j_{s-1}}$ up to a constant multiplicative error at any {\it node} $t<\tilde{\qubits}$.
    \item Compute the leaf weight $\tr(G^{(\tilde{\qubits})}_{j_1,\dots,j_{\tilde{\qubits}}}) \prod_{t=1}^{\tilde{\qubits} }w^{(t)}_{j_t|j_1,\dots,j_{t-1}}$ exactly.
    \item Sample according to the weights $w^{(t)}_{j_t|j_1,\dots,j_{t-1}}$. 
\end{enumerate}

For queries of type 3, we observe that the weights $w^{(t)}_{j_t|j_1,\dots,j_{t-1}}$ correspond to the product of the exponential distributions over $\ell$ and $\ell'$ in Eq.~\eqref{eq:firstdecompXicAX}, and  uniform distributions over $P,P'$ in the same equation, the seven terms from Eq.~\eqref{eq:lemm_full_expression_IpOj}  and two signs in Eq.~\eqref{lemm:plusminuspauli}.
After the truncation of $\ell,\ell'$ at $\ell_{\max}$ described above, there are $T=\text{poly}(\qubits/\varepsilon)$ terms, for which one can efficiently compute the weights $w^{\smash{(t)}}_{j_t|j_1,\dots,j_{t-1}}$ and thus sample from the resulting distribution. 

For queries of type 2, since $\tr (G^{(\smash{\tilde{\qubits})}}_{j_1,\dots,j_{\tilde{\qubits}}})$ is simply the normalization of the MPS in Eq.~\eqref{eq:normfactor}, we can  compute it efficiently.

Finally, we just need to explain how to answer queries of type 1. Let us fix $t$ and $j_1,\dots,j_t$, and we will omit the subscripts $j_1,\dots,j_t$ and superscripts $(t)$ henceforth. We need to compute $\tr(G)$ to constant multiplicative error, meaning that we are required to output an estimate $\widetilde{W}$ such that $\tr(G)c\leq \widetilde{W}\leq C\tr(G)$ for some positive constants $c,C$.  Note that
\begin{equation}
    \tr(G)=\tr(O e^{-\beta H_{[t+1,\qubits]}/2}(I+cB)e^{-\beta H_{[t+1,\qubits]}/2}),
\end{equation}
with 
\begin{equation}
    O=\tr_{[t]}\Big(M^\dagger_t \cdots M_1^\dagger M_1\cdots M_t (\sigma\otimes I_{[t+1,t+m-1]})  \Big),
\end{equation}
where $\tr_{[t]}$ denotes the partial trace over the first $t$ qubits.
The operator $O$ is positive semidefinite and is supported between qubits $[t+1,t+m-1]$, since the operators $M_1,\dots,M_t$ act between qubits $[1,t+m-1]$. Further, $O$ can be computed efficiently, since it is a contraction of a matrix product operator with the product state $\sigma$. Now, we leverage the following:
\begin{fact}If $J$ and $K$ are positive semidefinite matrices, and $J$ is invertible, then
\begin{equation}
    \frac{1}{\norm{J^{-1}}} \tr(K)\leq \tr(JK) \leq \norm{J} \tr(K).
\end{equation}
\end{fact}
Taking $K= e^{-\beta H_{[t+1,\qubits]}/2}Oe^{-\beta H_{[t+1,\qubits]}/2}$ and $J=I+cB$ and using that $|c|\leq 1/7$ in the nontrivial case $B\neq I$, we have that
\begin{equation}
    \frac{6}{7} W\leq \tr(G)\leq \frac{8}{7}  \, W,
    \label{eq:approxofGwithW}
\end{equation}
with $W=\tr(O \exp(-\beta H_{[t+1,\qubits]})).$ 

We reduced the problem to computing $W$ to constant multiplicative error. We may use standard MPO approximations for thermal Gibbs states in 1D to estimate $W$. For example, Molnar, Schuch, Verstraete and Cirac~\cite{Molnar2015} construct an MPO $\tilde{\rho}_\beta$ with bond dimension $((\qubits-t)/\delta)^{\bigO{\beta}}$ such that the trace norm $\norm{\exp(-\beta H_{[t+1,\qubits]}) - \tilde{\rho}}_1\leq \delta Z_{[t+1,\qubits]}$ with $Z_{[t+1,\qubits]}=\tr({\exp(-\beta H_{[t+1,\qubits]})})$. From here, we can estimate the partition function to multiplicative error
\begin{equation}
   \frac{\tr(\tilde{\rho})}{1+\delta}\leq  Z_{[t+1,\qubits]}\leq \frac{\tr(\tilde{\rho})}{1-\delta}
\end{equation}
and further $W$ to additive error
\begin{equation}
    W=\tr(\tilde{\rho}O)\pm \delta Z_{[t+1,\qubits]} \norm{O}.
\end{equation}
We can turn this into constant multiplicative error, since $O\geq 0$ is supported on $[t+1,t+m-1]$, which gives
\begin{align*}
    W&=\tr\big(O\,e^{-\beta H_{[t+1,\qubits]}}\big)
    \geq\norm{O}\,\min_{\ket{a}}\sum_{b}\bra{ab}e^{-\beta H_{[t+1,\qubits]}}\ket{ab}
    \geq\norm{O}\,\min_{\ket{a}}\sum_{b}e^{-\beta\bra{a b}H_{[t+1,\qubits]}\ket{ab}}\\
    &\geq\norm{O}\,e^{-\beta(m+\locality-1)}\sum_{b}e^{-\beta\bra{b}H_{[t+m,\qubits]}\ket{b}}=e^{-\beta(m+\locality-1)}Z_{[t+m,\qubits]}\norm{O}\geq 2^{-m}e^{-2\beta(m+\locality-1)}Z_{[t+1,\qubits]}\norm{O}.
\end{align*}
Here, the minimum is over unit vectors $\ket{a}$ on $[t+1,t+m-1]$, $\{\ket{b}\}$ is an eigenbasis of the qubits in $[t+m,\qubits]$, and $Z_{[t+m,\qubits]}=\tr(e^{-\beta H_{[t+m,\qubits]}})$. The first inequality holds because $W$ is the trace of $O\geq0$ against a positive operator $\tr_{[t+m,\qubits]}(e^{-\beta H_{[t+1,\qubits]}})$, so it is at least $\tr(O)\geq\norm{O}$ times the smallest eigenvalue of that operator. The second inequality is $\bra{v}e^{A}\ket{v}\geq e^{\bra{v}A\ket{v}}$. The third uses that $H_{[t+1,\qubits]}-H_{[t+m,\qubits]}$ consists of at most $m+\locality-1$ terms of norm at most $1$. The last one follows from Weyl's inequality for eigenvalues~\cite{Bhatia1997}, which gives $Z_{[t+1,\qubits]}\leq 2^{m-1}e^{\beta(m+\locality-1)}Z_{[t+m,\qubits]}$.
Thus, for $O\neq 0$,
\begin{equation*}
    \frac{|\tr(\tilde{\rho}O)-W|}{W}\leq\delta\,2^m e^{2\beta(m+\locality-1)}\leq\frac{1}{10}
\end{equation*}
upon choosing $\delta\leq 2^{-m}e^{-2\beta(m+\locality-1)}/10$. 
\ifbool{journal}{\end{widetext}}{}

%% file: acknowledgments.tex
The authors would like to thank Allen Liu for several valuable discussions in the early stages of this work. Part of this work was done while A.B. was visiting the Simons Institute. S.C. and S.P.C. are supported by the Center for Ultracold Atoms (an
NSF Physics Frontiers Center; PHY-2317134) and NSF CAREER (DMR-2237244). S.C. also acknowledges support from the Sloan Foundation through a Sloan Research Fellowship.